\documentclass[10pt]{article}

\usepackage{amssymb,amsmath}
\usepackage{mathrsfs,color}
\usepackage{amsthm}
\usepackage{mathtools}
\usepackage{cancel}
\usepackage{xspace}
\DeclareRobustCommand\mos[1]{\mathrel{|}\joinrel
\stackrel{#1}{\mathrel{=}}}
\usepackage[ruled,lined]{algorithm2e}% provides Algorithm environment

\DeclareRobustCommand\mfol{\mos{}_{{\ }_{FOL}}}

\DeclareRobustCommand\vds[1]{\,\mathrel{|}\joinrel\joinrel\joinrel
\frac{#1}{ \ \ \ }}

\makeatletter

\newcommand{\tr}{\sigma}

\DeclareRobustCommand\mos[1]{\mathrel{|}\joinrel
\stackrel{#1}{\mathrel{=}}}
\newcommand{\mosn}[1]{\not \mos{#1}}

\DeclareRobustCommand\vds[1]{\,\mathrel{|}\joinrel\joinrel\joinrel
\frac{#1}{ \ \ \ }}

\newcommand{\mb}{\scriptscriptstyle{\Box}}
\newcommand{\lra}{\leftrightarrow}
\newcommand{\vp}{\varphi}
\newcommand{\srb}{\sigma}

\usepackage{bm,listings,todonotes}

\newcommand{\te}[1]{\texttt{#1}}

\newtheorem{theorem}{Theorem}
\newtheorem{lemma}[theorem]{Lemma}

\newtheorem{remark}[theorem]{Remark}

\newtheorem{proposition}[theorem]{Proposition}
\newtheorem{definition}[theorem]{Definition}

  \title{Many-Sorted Hybrid Modal Languages}
  \author{Ioana Leu\c stean, 
	Natalia Moang\u a  and Traian Florin Șerbănuță \\
{\small Faculty of Mathematics and Computer Science, University of Bucharest,}\\
{\small Academiei nr.14, sector 1, C.P. 010014,  Bucharest, Romania} \\
{\small ioana@fmi.unibuc.ro}  \hspace*{0.3cm}  {\small natalia.moanga@drd.unibuc.ro}  \hspace*{0.3cm}  {\small traian.serbanuta@fmi.unibuc.ro }
} 

\date{}

\begin{document}

\maketitle

\begin{abstract}
We continue our investigation into hybrid polyadic multi-sorted logic with a focus on expresivity
related to the operational and axiomatic semantics of programming languages, and relations with first-order logic.
We identify a fragment of the full logic, for which we prove sound and complete deduction
and we show that it is powerful enough to represent both the programs and their semantics in an uniform way.
Although weaker than other hybrid systems previously developed, this system is expected to have better computational properties.
Finally, we provide a standard translation from full hybrid many-sorted logic to first-order logic.

Operational semantics, Program verification 

\end{abstract}

%\begin{keyword}
%Hybrid modal logic, Many-sorted logic, Standard Translation,\\Operational semantics, Program verification 
%\end{keyword}

\section{Introduction}\label{sec:intro}
This paper presents several hybrid modal logic systems based on the initial many-sorted structure we have developed~\cite{noi} and progressively incorporating different operators and binders to it. These findings have enabled us to bridge the gap between the full many-sorted polyadic modal logic and First-Order Logic by developing a standard translation between them.
 
In Section \ref{sec:intro} we recall our many-sorted polyadic modal logic, $\mathcal{K}_{\Sigma}$, introduced in \cite{noi}, by presenting all the necessary information:  the syntax, the semantics and the deductive system, in order for the reader to get familiarized with this logic. In Section \ref{sec:basichyb}, we propose and study $\mathcal{H}_{\Sigma}(@_z)$ a hybrid extension of $\mathcal{K}_{\Sigma}$ and we prove its soundness and completeness. Moreover, we provide an example of using this system to axiomatically express operational semantics and to derive proofs for statements concerning program executions. Sections \ref{sec:hyb} and \ref{sec2} recall two related hybrid systems introduced in~\cite{noi3}: $\mathcal{H}_{\Sigma}(\forall)$, an orthodogal extension of $\mathcal{K}_{\Sigma}$; and $\mathcal{H}_{\Sigma}(@_z,\forall)$, a common extension of both $\mathcal{H}_{\Sigma}(@_z)$ and $\mathcal{H}_{\Sigma}(\forall)$.  The paper concludes by providing a standard translation from $\mathcal{H}_{\Sigma}(@_z,\forall)$ to first order logic, showing that any many-sorted modal formula corresponds to a first-order formula from its corresponding first-order language.

\subsection{Preliminaries: a many-sorted polyadic modal logic}
\label{subsec:prel}
 For a general background on modal logic we refer to \cite{mod}. Basically, on top of modal logic we have added the sorts for each variable and the many-sorted polyadic operators $\sigma$ together with the corresponding relation. The polyadic operators are defined also in \cite{mod}, but in a mono-sorted version.
 
Our language is determined by a fixed, but arbitrary,  many-sorted signature ${\bf \Sigma}=(S, \Sigma)$ and an $S$-sorted set of  propositional variables $P=\{P_s\}_{s\in S}$ such that  $P_s\neq \emptyset$ for any $s\in S$ and  $P_{s_1} \cap P_{s_2} = \emptyset $ for any  $s_1 \neq s_2$ in $S$.
For any $n\in {\mathbb N}$ and $s,s_1,\ldots, s_n\in S$ we denote 
$\Sigma_{s_1\ldots s_n,s}=\{\sigma\in \Sigma\mid \sigma:s_1 \cdots s_n\to s\}$.

The set of formulas of ${\mathcal K}_{ \Sigma}$, the many-sorted polyadic modal logic  defined in \cite{noi}, is an $S$-indexed family  inductively defined by:

 \begin{center}
 $\phi_s ::=  p\,|\,\neg\phi_s\,|\,{\phi}_s\vee {\phi}_s\,|\, \sigma ({\phi}_{s_1}, \ldots , {\phi}_{s_n} )$
\end{center}

where $s\in S$, $p\in P_s$ and  $\sigma \in \Sigma_{s_1 \cdots s_n,s}$.

We use the classical definitions of the derived logical connectors: for any $\sigma\in \Sigma_{s_1 \ldots s_n,s}$ the {\em dual operation} is
$\sigma^{\mb} (\phi_1, \ldots , \phi_n ) := \neg\sigma (\neg\phi_1, \ldots , \neg\phi_n ).$

In the sequel,  by $\phi_s$ we mean that $\phi$ is a formula of sort $s\in S$. Similarly, $\Gamma_s$ means that $\Gamma$ is a set of formulas of sort $s$. When the context uniquely determines  the sort of a state symbol, we shall omit the subscript. 

The deductive system is presented in Figure \ref{fig:k1}.

\begin{figure}[h]
\centering

{\bf The system   ${\mathcal K}_{ \Sigma}$}

\begin{itemize}
\item For any $s\in S$, if $\phi$ is a formula of sort $s$ which is a theorem in propositional logic, then $\phi$ is an axiom. 
\item Axiom schemes: for any $\sigma\in \Sigma_{s_1\cdots s_n,s}$ and for any formulas $\phi_1,\ldots, \phi_n,\phi,\chi$ of appropriate sorts, the following formulas are axioms:

\begin{tabular}{rl}
$(K_\sigma)$ & $\sigma^{\mb}(\ldots,\phi_{i-1},\phi\rightarrow\chi,\phi_{i+1}, \ldots)\to$\\ &\hspace*{0.5cm}$( \sigma^{\mb}(\ldots ,\phi_{i-1}, \phi, \phi_{i+1},\ldots) \to \sigma^{\mb}(\ldots ,\phi_{i-1}, \chi, \phi_{i+1},\ldots))$\\
$(Dual_\sigma)$& $\sigma (\psi_1,\ldots ,\psi_n )\leftrightarrow \neg \sigma^{\mb} (\neg \psi_1,\ldots ,\neg \psi_n )$
\end{tabular}

\item Deduction rules: {\em Modus Ponens} and {\em Universal Generalization}
\medskip

\begin{tabular}{rl}
$(MP)$ & if $\vds{s}\phi$ and $\vds{s}\phi\to \psi$ then 
$\vds{s}\psi$\\
$(UG)$ &  if $\vds{s_i}{\phi}$ then $\vds{s}\sigma^{\mb} (\phi_1, .. ,\phi, ..\phi_n)$
\end{tabular}
\end{itemize} 
\caption{$(S,\Sigma)$ modal logic}\label{fig:k1}
\end{figure}

In order to define the semantics we introduce $(S,\Sigma)$\textit{-frames} and $(S,\Sigma)$\textit{-models}. An $(S,\Sigma)${\em -frame} is a tuple 
 $\mathcal{F} =({W},(R_\sigma)_{\sigma\in\Sigma})$
 such that:
\begin{itemize}
\item  ${W} =\{ W_s \}_{s\in S}$ is an  $S$-sorted set of worlds and $W_s\neq \emptyset$ for any $s\in S$,
\item  ${R}_{\sigma} \subseteq  W_s \times W_{s_1} \times \ldots \times W_{s_n}$  for any $\sigma \in \Sigma_{s_1 \cdots s_n,s}$.
\end{itemize}

An $(S,\Sigma)$-{\em model based on} $\mathcal{F}$ is a pair  ${\mathcal M}= ({\mathcal F},V)$ where $V =\{V_s\}_{s\in S}$ such that $V_s : P_s \to \mathcal{P}(W_s)$ for any $s\in S$. The model $\mathcal{M}= (\mathcal{F},V)$ will be simply denoted as  $\mathcal{M}= ({W}, (R_\sigma)_{\sigma\in\Sigma},V)$. For $s\in S$, $w\in W_s$ and  $\phi$ a formula of sort $s$, the  many-sorted  \textit{satisfaction relation} $\mathcal{M},w\mos{s} \phi$
is inductively defined as follows:
\begin{itemize}
\item $\mathcal{M},w \mos{s} p$ iff $w\in V_s(p)$
\item $\mathcal{M},w \mos{s} \neg \psi$ iff $\mathcal{M},w \mosn{s}\psi$
\item $\mathcal{M},w \mos{s} \psi_1 \vee \psi_2$ iff $\mathcal{M},w \mos{s} \psi_1$ or $\mathcal{M},w \mos{s} \psi_2$ 
\item if $\sigma\in \Sigma_{s_1 \ldots s_n,s} $,  then $\mathcal{M},w \mos{s} \sigma(\phi_1, \ldots , \phi_n )$ iff  for any $i \in [n]$ there exist $w_i \in W_{s_i}$ such that ${R}_{\sigma} ww_1\ldots w_n$ and
 $\mathcal{M},w_i  \mos{s_i} \phi_i$.
\end{itemize}

\begin{definition}[Validity and satisfiability] Let $s\in S$ and assume $\phi$ is a formula of sort $s$. Then $\phi$ is {\em satisfiable} if  ${\mathcal M},w\mos{s}\phi$ for some model $\mathcal M$ and some $w\in W_s$.  The formula $\phi$ is {\em valid} in a model $\mathcal M$ if ${\mathcal M},w\mos{s}\phi$ for any $w\in W_s$; in this case we write ${\mathcal M}\mos{s}\phi$. The formula $\phi$ is valid in a  frame $\mathcal F$  if $\phi$ is valid in all the models based on $\mathcal F$; in this case we write ${\mathcal F}\mos{s}\phi$. Finally, the formula $\phi$ is valid if $\phi$ is valid in all frames; in this case we write $\mos{s}\phi$. 
\end{definition}

The {\em set of theorems} of ${\mathcal K}_{ \Sigma}$ is the least set of formulas that contains all the axioms and it is closed under deduction rules. Note that the set of theorems is obviously closed under {\em $S$-sorted uniform substitution} (i.e. propositional variables of sort $s$ are uniformly replaced by formulas of the same sort). If $\phi$ is a theorem of sort $s$ write \mbox{$\vds{s}_{{\mathcal K}_{ \Sigma}}\phi$}, or simply \mbox{$\vds{s} \phi$.} Obviously, ${\mathcal K}_{ \Sigma}$ is a  generalization of the modal system $\mathbf K$ (see \cite{mod} for the mono-sorted version). The completeness theorem of ${\mathcal K}_{ \Sigma}$ is proved in \cite{noi}.  

\section{The many-sorted  basic hybrid  modal logic ${\mathcal H}_{ \Sigma}(@_z)$}\label{sec:basichyb}

        Let $(S,\Sigma)$ be a many-sorted signature. A basic hybrid modal logic is defined on top of modal logic $\mathcal{K}_{\Sigma}$ by  adding {\em nominals}, {\em states variables} and specific operators.  Nominals  allow us to directly refer the worlds (states) of a model, since they are evaluated in singletons in any model. However, a nominal may refer different worlds in different models. The sorts will be denoted by $s$, $t$, $\ldots$ and by ${\rm PROP}=\{{\rm PROP}_s\}_{s\in S}$, ${\rm NOM}=\{{\rm NOM}_s\}_{s\in S}$ and ${\rm SVAR}=\{{\rm SVAR}_s\}_{s\in S}$ we will denote some countable $S$-sorted sets. The elements of ${\rm PROP}$ are ordinary propositional variables and they will be denoted $p$, $q$,$\ldots$; the elements of ${\rm NOM}$ are called {\em nominals} and they will be denoted by $j$, $k$, $\ldots$; the elements of ${\rm SVAR}$ are called {\em state variables} and they are denoted $x$, $y$, $\ldots$.  We shall assume that for any distinct sorts $s\neq t\in S$, the corresponding sets of propositional variables, nominals and state variables are distinct. A {\em state symbol} is a nominal or a state variable.
        
%We recall that hybrid logics are modal logics that have special symbols (called ``nominals") that name the particular states of a model. 
Recall that the satisfaction in modal logic is local, i.e. one analyzes what happens in a given point of the model. With respect to this, nominals can be seen as local constants and, given a model (a frame and an evaluation), the value of a nominal is a fixed singleton set. State variables are variables that range over the individual points of a model, while the usual (propositional) variables range over arbitrary sets of points. 
        
For this section we drew our inspiration mainly from \cite{hybtemp}. As already announced,  in this section we extend the system defined in Section~\ref{sec:intro} by adding the satisfaction operators $@_z^s$ where $s\in S$ and $z$ is a {\em state symbol}. The formulas 
of  ${\mathcal H}_{ \Sigma}(@_z)$ are defined as follows:
\begin{center}
$\phi_s :=  p\mid j\mid x_s\mid \neg \phi_s \mid\phi_s \vee \phi_s \mid \sigma(\phi_{s_1}, \ldots, \phi_{s_n})_s
  \mid @_z^s\psi_t$
\end{center}

Here, $p\in {\rm PROP}_s$, $j\in {\rm NOM}_s$, $t\in S$, $x\in {\rm SVAR}_s$, $\sigma\in \Sigma_{s_1\cdots s_n,s}$, $z$ is a state symbol of sort $t$ and 
$\psi$ is a formula of sort $t$.

In order to define the semantics for ${\mathcal H}_{ \Sigma}(@_z)$ more is needed. Given an $(S, \Sigma)$-model ${\mathcal M}=(W, (R_\sigma)_{\sigma\in\Sigma}, V)$, an {\em assignment} is an $S$-sorted function \mbox{$g : {\rm SVAR} \rightarrow W$,} which evaluates states variables to singleton sets, and for any $s \in S$ we have $g_s: {\rm SVAR}_s \to W_s$.

The satisfaction relation is defined similar with the one in ${\mathcal K}_{ \Sigma}$, but we only need to add the definition for $@_z$: \begin{center}
$\mathcal{M},g,w \mos{s} @_z^s\phi$ if and only if $\mathcal{M},g,Den_g(z)\mos{t} \phi$ 
\end{center}
where $z$ is a state symbol of sort $t$ and $\phi$ is a formula of the same sort $t$. Here, $Den_g(z)$ is the denotation of the state symbol $z$ of sort $s$ in an $(S, \Sigma)$-model $\mathcal{M}$ with an assignment function $g$, where $Den_g(z)=V_s(z)$ if $z$ is a nominal, and $Den_g(z)=g_s(z)$ if $z$ is a state variable. 

Let us remark that if $z$ is a nominal, then the satisfaction relation is equivalent with the one in \cite{noi2}:\begin{center}
 $\mathcal{M},g,w \mos{s} @_z^s\phi$ if and only if $\mathcal{M},g,Den_g(z)\mos{t} \phi$ if and only if\\ $\mathcal{M}, g, v \mos{t} \phi$ where $Den_g(z)=V_t(z)=\{v\}$. 
\end{center}

One important remark is the definition of the satisfaction modalities: 
{\em if $z$ and $\phi$ are a state symbol and a formula both of the sort $t\in S$, then we define a family of satisfaction operators 
$\{@_z^s\phi\}_{s\in S}$} such that $@_z^s\phi$ is a formula of sort $s$ for any $s\in S$. This means that $\phi$ is true at the world denoted by $z$ on the sort $t$ and is acknowledged on any sort $s\in S$. For example, if we take $j$ and $k$ two nominals of sort $t$ and $s\neq t$ the formula $@^s_j\neg k$ expresses the fact that at any world of sort $s$ we know that the worlds of sort $t$ named by $j$ and $k$ are different. So, our sorted worlds are not isolated any more, both from a syntactic and a semantic point of view.

\begin{figure}[!h]
\centering

{\bf The system} ${\mathcal H}_{\Sigma}(@_z)$
\begin{itemize}
\item The axioms and the deduction rules of ${\mathcal K}_{\Sigma}$

\item Axiom schemes: any formula of the following form is an axiom, where $s,s',t$ are sorts,  $\sigma\in\Sigma_{s_1\cdots s_n,s}$,  $\phi,\psi, \phi_1,\ldots,\phi_n$ are formulas (when necessary, their sort is marked as a subscript), and $y$, $z$ are state symbols:

$\begin{array}{rl}
(K@) & @_z^{s} (\phi_t \to \psi_t) \to (@_z^s \phi \to @_z^s \psi) \\
(SelfDual) & @^s_z \phi_t \leftrightarrow \neg @_z^s \neg \phi_t \\
(Intro)  & z \to (\phi_s \leftrightarrow @_z^s \phi_s)\\
(Agree) &  @_y^{t}@_z^{t'} \phi_s \leftrightarrow @^t_z \phi_s\\
(Ref) & @_z^sz_t \\
(Back) & \sigma(\ldots,\phi_{i-1}, @_z^{s_i} {\psi}_t,\phi_{i+1},\ldots)_s\to @_z^s {\psi}_t 
\end{array}$\\
\medskip

%$\begin{array}{rl}
%(Q1) & \forall x\,(\phi \to \psi) \to (\phi \to \forall x\, \psi)
%  \mbox{ where $\phi$ contains no free occurrences of x}\\
%(Q2) & \forall x\,\phi \to \phi[y\slash x] 
%  \mbox{ where $y$ is substitutable for $x$ in $\phi$}\\
%(Name) & \exists x \, x \\ 
%(Barcan) & \forall x \,\sigma^{\mb}(\phi_1, \ldots, \phi_n) \to \sigma^{\mb}(\phi_1, \ldots,\forall x\phi_{i},\ldots, \phi_n)\\
%(Barcan@) & \forall x \, @_z\phi \to @_z \forall x\,\phi, \mbox{where } x \neq z\\
% (Nom\, x) & @_z x\wedge @_y x \to @_z y 
%\end{array}$
%
%\medskip

\item Deduction rules:

\begin{tabular}{rl}
$(BroadcastS)$ & if $\vds{s}@_z^s\phi_t$ then $\vds{s'}@_z^{s'}\phi_t$ \\
 $(Gen@)$& if $\vds{s'} \phi$ then $\vds{s} @_z \phi$, where $z $ and $\phi$ have the same sort $s'$\\
%$(Name@)$ & if $\vds{s} @_j \phi$ then $\vds{s'} \phi$, where $j$ does not occur in $\phi$\\
$(Paste0)$ & if $\vds{s} @^s_z (y \wedge \phi) \to \psi$ then $\vds{s} @_z \phi \to \psi$\\ & where $z$ is distinct from $y$ that does not occur in $\phi$ or $\psi$\\
$(Paste1)$ & if $\vds{s} @^s_z \sigma(\ldots, y \wedge \phi,\ldots)  \to \psi$ then $\vds{s} @^s_z \sigma(\ldots, \phi, \ldots) \to \psi$\\ & where $z$ is distinct from $y$ that does not occur in $\phi$ or $\psi$\\
%$(Gen)$ & if $\vds{s}\phi$ then $\vds{s}\forall x\phi$, where $\phi\in Form_s$ and $x\in {\rm SVAR}_t$ for some $t\in S$.
\end{tabular}
\end{itemize}
\caption{$(S,\Sigma)$ basic hybrid modal logic}\label{fig:unu}
\end{figure}

\begin{proposition}[Soundness] The deductive systems for ${\mathcal H}_{ \Sigma}(@_z)$ from Figure \ref{fig:unu} is sound.
\end{proposition}
\begin{proof}
Let $\mathcal{M}$ be an arbitrary model and $w$ any state of sort $s$.

$(K_@)$
Suppose $\mathcal{M},g,w \mos{s} @_z^{s} (\phi_t \to \psi_t) $ if and only if $\mathcal{M},g,Den_g(z) \mos{t} \phi_t \to \psi_t$ if and only if $\mathcal{M},g, Den_g(z) \mos{t} \phi_t$ implies $\mathcal{M},g, Den_g(z) \mos{t} \psi_t$. Let us prove the non-trivial case: suppose that $\mathcal{M},g, w  \mos{s} @_j^s \phi_t$. Then $\mathcal{M},g, Den_g(z)  \mos{t} \phi_t$, but this implies that $\mathcal{M},g, Den_g(z) \mos{t} \psi_t$ if and only if $\mathcal{M},g, w  \mos{s} @_z^s \psi_t$. Therefore, $\mathcal{M},g, w \mos{s} @_z^s \phi_t \to @_z^s \psi_t$.

$(Agree)$
Suppose $\mathcal{M},g,w  \mos{t} @_y^{t}@_z^{t'} \phi_s $ if and only if $\mathcal{M},g,Den_g(y)  \mos{t'} @_z^{t} \phi_s $ implies $\mathcal{M},g,Den_g(z)  \mos{s}  \phi_s $. It follows that $\mathcal{M},g,w  \mos{t}  @^{t}_z \phi_s $.

$(SelfDual)$ 
Suppose $\mathcal{M},g,w  \mos{s} \neg @^s_z  \neg \phi_t $ if and only if $\mathcal{M},g,w  \mosn{s} @^s_z \neg  \phi_t$ if and only if  $\mathcal{M},g,Den_g(z)  \mosn{t}   \neg \phi_t$ if and only if $ \mathcal{M},g, Den_g(z) \mos{t} \phi_t$ if and only if $\mathcal{M},g,w  \mos{s}  @^s_z  \phi_t$.

$(Back)$
Suppose $\mathcal{M},g,w \mos{s}  \sigma(\ldots,\phi_{i-1}, @_z^{s_i} {\psi}_t,\phi_{i+1},\ldots)_s$ if and only if there is $(w_1,\ldots,w_n) \in W_{s_1}\times\cdots\times W_{s_n}$ such that  $R_{\sigma} ww_1\ldots w_n$ and $\mathcal{M},g,w_i  \mos{s_i} \phi_i$ for any $i \in [n]$. This implies that there is $w_i \in W_{s_i}$ such that  $\mathcal{M},g,w_i  \mos{s_i} @_z^{s_i} {\psi}_t $, so $ \mathcal{M},g,Den_g(z)  \mos{t} \psi_t $. Hence, $\mathcal{M},g,w \mos{s}  @_z^{s} {\psi}_t $

$(Ref)$ 
Suppose $\mathcal{M},g, w \mosn{s}@_z^s z_t$. Then $\mathcal{M},g, Den_g(z) \mosn{t} z$, contradiction.

$(Intro)$
Suppose $\mathcal{M},g,w \mos{s} z$  and  $\mathcal{M},g,w \mos{s}  \phi_s$. Then ${w}= Den_g(z)$, so we get that $\mathcal{M},g,Den_g(z) \mos{s} z$  and  $\mathcal{M},g,Den_g(z) \mos{s}  \phi_s$ implies that $\mathcal{M},g,w \mos{s} @_z^s \phi_s$.

Now, suppose $\mathcal{M},g,w \mos{s} z$ and $\mathcal{M},g,w \mos{s} @_z^s \phi_s$. Because from the first assumption we have $Den_g(z)=\{ w\}$, then, from the second one, we can conclude that $\mathcal{M},g,w\mos{s} \phi_s $.
\end{proof}

The following lemma generalizes the results from \cite{pureax}, being essentially used in the proof of the completeness theorem.

\begin{lemma}\label{lemma:bridge}
The following formulas are theorems:
\vspace*{3mm}

\begin{tabular}{ll}
$(Nom_z)$ & $@_z^s y_t\to (@_z^s\phi_t\leftrightarrow @_y^s\phi_t)$\\
& for any $s,t\in S$,  $z_t,y_t$ state symbols of sort $t$ and $\phi_t$ a formula \\
& of sort $t$. \\ 
$(Sym)$ &  $@_z^s y_t \to @_y^s z_t$ \\
  & where $s,t\in S$ and $z_t,y_t$ are state symbols of sort $t$,\\
$(Bridge)$ & $\sigma(\ldots \phi_{i_1}, z_{s_i} ,\phi_{i+1} \ldots) \wedge @_z^s \phi_{s_i} \to \sigma(\ldots\phi_{i-1}, \phi_{s_i} ,\phi_{i+1}, \ldots)$ \\
 & if $\sigma\in \Sigma_{s_1\ldots s_n,s}$, 
 $z_{s_i}$ is a state symbol of sort $s_i$ and $\phi_{s_i}$ is a\\
 & formula of sort $s_i$.
 \end{tabular}
% \medskip
%\item if $\vds{s}\phi\to j$ then 
%$\vds{t}\sigma(\ldots,\phi,\ldots)\to \sigma(\ldots,j,\ldots)\wedge @_j^t\phi$
%
%for any $s,t\in S$, $\sigma\in\Sigma_{t_1\cdots t_n,t}$, $j\in {\rm NOM}_s\cup{\nom}_s$ and $\phi$ a formula of sort $s$. 
% 
% \end{enumerate}
\end{lemma}

\begin{proof}
In the sequel, by PL we mean classical propositional logic and by ML we mean the basic modal logic.
\medskip

$(Nom_z)$ 
 
\noindent$(1)$ $ \vds{t} y_t \to (\phi_t \leftrightarrow @_y^t \phi_t) $\hfill $(Intro)$

\noindent$(2)$ $ \vds{s} @_z^s (y_t \to (\phi_t \leftrightarrow @_y^t \phi_t))$ \hfill  $(Gen@)$

\noindent$(3)$ $ \vds{s} @_z^s (y_t \to (\phi_t \leftrightarrow @_y^t \phi_t)) \to (  @_z^s y_t \to  @_z^s (\phi_t \leftrightarrow @_y^t \phi_t)) $ \hfill  $(K@)$

\noindent$(4)$  $\vds{s}  @_z^s y_t \to  @_z^s (\phi_t \leftrightarrow @_y^t \phi_t) $ \hfill  $(MP):(2),(3)$

\noindent$(5)$  $\vds{s}  @_z^s (\phi_t \leftrightarrow @_y^t \phi_t) \leftrightarrow (@_z^s \phi_t \leftrightarrow @_z^s @_y^t \phi_t)$  \hfill ML

\noindent$(6)$  $\vds{s}  @_z^s y_t \to (@_z^s \phi_t \leftrightarrow @_z^s @_y^t \phi_t)$\hfill PL:$(4),(5)$

\noindent$(7)$  $\vds{s}  @_z^s @_y^t \phi_t  \leftrightarrow  @_y^s \phi_t$ \hfill $(Agree)$

\noindent$(8)$  $\vds{s}  @_z^s y_t \to (@_z^s \phi_t \leftrightarrow @_y^s \phi_t)$\hfill PL:$(6),(7)$

\bigskip

$(Sym)$

\noindent$(1)$ $ \vds{s} @_y^s z_t \wedge @_z^s y_t\to @_z^s y_t$ \hfill $Taut$\\
\noindent$(2)$ $ \vds{s} (@_y^s z_t \wedge @_z^s y_t \to @_z^s y_t) \to (@_y^s z \to( @_z^s y_t\to @_z^s y_t))$ \hfill $Taut$\\
\noindent$(3)$ $ \vds{s} @_y^s z \to( @_z^s y_t \to @_z^s y_t)$ \hfill $(MP):(1),(2)$\\
\noindent$(4)$ $ \vds{s}( @_z^s y_t \to @_z^s y_t)\to @_z^s y_t$ \hfill  PL\\
\noindent$(5)$ $ \vds{s} @_y^s z \to @_z^s y_t$  \hfill PL\\
\noindent$(6)$ $ \vds{s} @_z^s y_t \to @_y^s z$  \hfill Analogue\\
\noindent$(7)$ $ \vds{s} @_z^s y_t \leftrightarrow @_y^s z$\hfill PL:$(5),(6)$
 
   \bigskip

$(Bridge)$ 

\noindent$(1)~$  $\vds{s} \sigma(\ldots \phi_{i-1}, z_{s_i} ,\phi_{i+1} \ldots) \wedge \sigma^{\mb}(\ldots,\neg\phi_{i-1}, \neg \phi_{s_i} ,\neg\phi_{i+1}, \ldots)\to \hfill $

\noindent\hfill$ \to \sigma (\ldots\phi_{i-1},z_{s_i} \wedge \phi_{s_i} ,\phi_{i+1}, \ldots)$ $\ \  $ ~ML

\noindent$(2)~$  $\vds{s_i}z_{s_i} \wedge \neg \phi_{s_i} \to @_z^{s_i} \neg \phi_{s_i} $ \hfill $(Intro)$

\noindent$(3)~$  $\vds{s}\sigma (\ldots\phi_{i-1},z_{s_i} \wedge \neg \phi_{s_i} ,\phi_{i+1}, \ldots) \to \sigma (\ldots\phi_{i-1}, @_z^{s_i} \neg \phi_{s_i} ,\phi_{i+1}, \ldots)$ \hfill ML

\noindent$(4)~$  $\vds{s} \sigma (\ldots\phi_{i-1}, @_z^{s_i} \neg \phi_{s_i} ,\phi_{i+1}, \ldots) \to @_z^s \neg \phi_{s_i}$ \hfill  $(Back)$

\noindent$(5)~$  $\vds{s}\sigma (\ldots\phi_{i-1},z_{s_i} \wedge \neg \phi_{s_i} ,\phi_{i+1}, \ldots) \to @_z^s \neg \phi_{s_i}$ \hfill PL:(3),(4)

\noindent$(6)~$  $ \vds{s} \sigma(\ldots \phi_{i_1}, z_{s_i},\phi_{i+1} \ldots) \wedge \sigma^{\mb}(\ldots,\neg\phi_{i-1}, \neg \phi_{s_i} ,\neg\phi_{i+1}, \ldots)\to  @_z^s  \neg \phi_{s_i}$

$\, $ \hfill PL:(1),(5)

\newpage
\noindent$(7)~$   $ \vds{s} \sigma(\ldots \phi_{i_1}, z_{s_i},\phi_{i+1} \ldots) \to( \sigma^{\mb}(\ldots,\neg\phi_{i-1}, \neg \phi_{s_i} ,\neg\phi_{i+1}, \ldots)\to  @_z^s  \neg \phi_{s_i})$ 

$ \, $ \hfill PL

\noindent$(8)~$   $ \vds{s} \sigma(\ldots \phi_{i_1}, z_{s_i},\phi_{i+1} \ldots) \to( \neg  @_z^s  \neg \phi_{s_i} \to \neg \sigma^{\mb}(\ldots,\neg\phi_{i-1}, \neg \phi_{s_i} ,\neg\phi_{i+1}, \ldots))$ 

$\, $ \hfill PL

\noindent$(9)~ $  $ \vds{s} \sigma(\ldots \phi_{i_1}, z_{s_i},\phi_{i+1} \ldots) \to(   @_z^s   \phi_{s_i} \to \sigma(\ldots,\phi_{i-1},  \phi_{s_i} ,\phi_{i+1}, \ldots))$ 

$ \ $ \hfill $(Dual),(SelfDual)$

\noindent$(10)$  $ \vds{s} \sigma(\ldots \phi_{i_1}, z_{s_i},\phi_{i+1} \ldots) \wedge   @_z^s   \phi_{s_i} \to \sigma(\ldots,\phi_{i-1},  \phi_{s_i} ,\phi_{i+1}, \ldots)$ 
\hfill PL

\end{proof}

\begin{lemma}\label{lem:prop}

Let $\Gamma_s$ be a maximal consistent set that contains a state symbol of sort $s$, and for all state symbols $z$, let $\Delta_z=\{ \phi \mid @_z^s \phi \in \Gamma_s \}$. Then:
\begin{itemize}
\item[1)] For every state symbol $z$ of sort $s$, $\Delta_z$ is a maximal consistent set that contains $z$.
\item[2)] For all state symbols $z$ and $y$ of same sort, $@^s_z \phi \in \Delta_y$ if and only if $@^s_z \phi \in \Gamma_s$.
\item[3)] There is a state symbol $z$ such that $\Gamma_s = \Delta_z$.
\item[4)] For all state symbols $z$ and $y$ of same sort, if $z \in \Delta_y$ then $\Delta_z= \Delta_y$.
\end{itemize}
\end{lemma}

\begin{proof}

\begin{itemize}
\item[1)] Recall that for any state symbol $z$ we have the $(Ref)$ axiom, so $@_z^s z_t \in \Gamma_s$. Hence, $z \in \Delta_z$. But, is $\Delta_z$ a consistent set? Let us suppose that is not. So there are $\chi_1, \ldots, \chi_n \in \Delta_j$ such that $\cancel{\vds{t}} \chi_1 \wedge \cdots \wedge\chi_n $, then ${\vds{t}} \neg( \chi_1 \wedge \cdots \wedge\chi_n) $. By use of $(Gen@)$ rule we get ${\vds{s}} @_z^s\neg( \chi_1 \wedge \cdots \wedge\chi_n) $, so $@_z^s\neg( \chi_1 \wedge \cdots \wedge\chi_n) \in \Gamma_s$. By $(SelfDual)$ axiom, we get that $\neg @_z^s( \chi_1 \wedge \cdots \wedge\chi_n) \in \Gamma_s$. But on the other hand, if $\chi_1, \ldots, \chi_n \in \Delta_j$, then by definition of $\Delta_j$ we have that $@^s_z \chi_1 , \ldots, @^s_z \chi_n \in \Gamma_s$, and because $@^s_z$ is a normal modality, then $@^s_z (\chi_1 \wedge \cdots \wedge\chi_n) \in \Gamma_s $ as well. But this contradicts the consistency of $\Gamma_s$. Therefore $\Delta_z$ is consistent. 

Now, let us check if $\Delta_z$ is maximal. Assume it is not. Then there is a formula $\chi$ of sort t such that $\chi \not \in \Delta_z$ and $\neg \chi \not \in \Delta_z$. But then $@_z^s \chi \not \in \Gamma_s$ and $@_z^s \neg \chi \not \in \Gamma_s$. But also $\Gamma_s $ is a maximal consistent set, then $\neg @_z^s \chi \in \Gamma_s$ and $\neg @_z^s \neg \chi \in \Gamma_s$. On the other hand, if $\neg @_z^s  \neg \chi \in \Gamma_s$, then by $(SelfDual)$ axiom we get that $@_z^s  \chi\in \Gamma_s$, and this contradicts the consistency of $\Gamma_s$. Hence, we conclude that $\Delta_z$ is a maximal consistent set.

\item[2)] By definition of $\Delta_y$, $@_z^t \phi \in \Delta_y$ if and only if $@_y^s @_z^t \phi \in \Gamma_s$. By $(Agree)$ axiom we have that  $@_y^s @_z^t \phi \in \Gamma_s$ if and only if  $ @_z^s \phi \in \Gamma_s$. This is called the $@$-agreement property, which it plays an important role in the completeness proof.

\item[3)] Let the state symbol $z$ of sort $s$ be contained in $\Gamma_s$. Suppose $\phi \in \Gamma_s$. Because $z \in \Gamma_s$, by $(Intro)$ axiom we get $@_z^s \phi \in \Gamma_s$, and by definition of $\Delta_z$, we have $\phi \in \Delta_z$. Conversely, if $\phi \in \Delta_z$, then by definition of $\Delta_z$ it follows that $@_z^s \phi \in \Gamma_s$. Moreover, $z \in \Gamma_s$ and using again the same axiom we get that $ \phi \in \Gamma_s $.

\item [4)] Let $z \in \Delta_y$, then by definition of $\Delta_y$ we have that $@^s_y z \in \Gamma_s$ and by $(Sym)$ we get that $@^s_z y \in \Gamma_s$. Firstly, let us prove that $\Delta_y \subseteq \Delta_z$. Let $\phi \in \Delta_y$, then by definition of $\Delta_y$ we have that $@^s_y \phi \in \Gamma_s$. Also,  $@^s_z y \in \Gamma_s$, so by $(Nom_z)$ it follows that $@^s_z \phi \in \Gamma_s$ and hence that $\phi \in \Delta_z$. Secondly, a similarly $(Nom_z)$-based proof shows that $\Delta_z \subseteq \Delta_y$. 
\end{itemize}
\end{proof}

This Lemma gives us the maximal consistent sets needed in the Existence Lemma. We build our models out of named sets, i.e. sets containing nominals. But more is needed in order for our model to support an Existential Lemma. Therefore, we add the $Paste$ rules, as you can see in Figure \ref{fig:unu}. In this setting, the system is still sound as we prove in the following:

$(BroadcastS)$ 
Suppose $\mathcal{M},g,w \mos{s} @_z^s  \phi_t$  if and only if $\mathcal{M},g,Den_g(z) \mos{t} \phi_t$. Hence, for any $s' \in S$ we have $\mathcal{M},g,w \mos{s'} @_z^{s'}  \phi_t$.

Now, let $\mathcal{M}$ be an arbitrary named model.

%$(Name @)$ Suppose $\mathcal{M},g,w \mos{s} @_j^s  \phi$ if and only if $\mathcal{M},g,v \mos{s'} \phi$ where $V^{\nom}_{s'}(j)=\{v \}$, but we work in named models, therefore, in any model $\mathcal{M}$ there exist $v$ and $j$ where $V^{\nom}_{s'}(j)=\{v \}$ and this implies $\mathcal{M},g,v \mos{s'} \phi$.
%
$(Paste0)$ Suppose $\mathcal{M},g,w \mos{s} @_z^s(y \wedge \phi) \to \psi$ if and only if $\mathcal{M},g,w \mos{s} @_z^s(y \wedge \phi)$ implies $\mathcal{M},g,w \mos{s} \psi$. Hence, ($\mathcal{M},g,v\mos{s'} y \wedge \phi$ where $Den_g(z)=\{v \}$ implies  $\mathcal{M},g,w \mos{s} \psi$) if and only if ($\mathcal{M},g,v \not\!\!\mos{s'} y$ and $\mathcal{M},g,v\not\!\!\mos{s'} \phi$, where $Den_g(z)=\{v \}$, or  $\mathcal{M},g,w \mos{s} \psi$). It follows that ($\mathcal{M},g,v \not\!\!\mos{s'} y$ or  $\mathcal{M},g,w \mos{s} \psi$) and ($\mathcal{M},g,v\not\!\!\mos{s'} \phi$ or  $\mathcal{M},g,w \mos{s} \psi$), where $Den_g(z)=\{v \}$. Then, ($\mathcal{M},g,v\not\!\!\mos{s'} \phi$ or  $\mathcal{M},g,w \mos{s} \psi$), where $Den_g(z)=\{v \}$. So, $\mathcal{M}, g, w \mos{s} @^s_z \phi \to \psi$.

$(Paste1)$ Suppose $\mathcal{M},g,w \mos{s} @_z^s  \sigma(\psi_1, \ldots, \psi_{i-1}, y \wedge \phi, \psi_{i+1}, \ldots, \psi_n) \to \psi$ if and only if $\mathcal{M},g,w \mos{s} @_z^s  \sigma(\psi_1, \ldots,$ $\psi_{i-1},  y \wedge \phi, \psi_{i+1}, \ldots, \psi_n)$ implies $\mathcal{M},g,w \mos{s} \psi$. Hence,  $\mathcal{M},g,v\mos{s'}   y \wedge \phi$ where $Den_g(z)=\{v \}$ if and only if exists $(v_1, \ldots,v_n) \in W_{s_1}\times \ldots\times W_{s_n}$ such that $R_{\sigma}v v_1 \ldots v_i \ldots v_n$ where $Den_g(z)=\{v \}$ and $\mathcal{M},g,v_e \mos{s'}  \psi_e$ for any $e \in [n], e\neq i$ and $  \mathcal{M}, g, v_i \mos{s_i} y \wedge \phi$. Hence,  $  \mathcal{M}, g, v_i \mos{s_i} y $ and  $  \mathcal{M}, g, v_i \mos{s_i}  \phi$, so $Den_g(y)=\{v_i\}$ and  $  \mathcal{M}, g, v_i \mos{s_i} \phi$. Then, if there exists $(v_1, \ldots,v_n) \in W_{s_1}\times \ldots\times W_{s_n}$ such that $R_{\sigma}v v_1 \ldots v_i \ldots v_n$ where $Den_g(z)=\{v \}$ and  $\mathcal{M},g,v_e \mos{s'}  \psi_e$ for any $e \in [n], e\neq i$ and $  \mathcal{M}, g, v_i \mos{s_i} \phi$, these imply $\mathcal{M},g,w \mos{s} \psi$. So, $\mathcal{M},g,v \mos{s'}  \sigma(\psi_1, \ldots, \psi_{i-1}, \phi, \psi_{i+1}, \ldots, \psi_n) $ where $ Den_g(z)=\{v \}$ implies $\mathcal{M},g,w \mos{s} \psi$. In conclusion, $\mathcal{M},g,w \mos{s'}  @_z^s\sigma(\psi_1, \ldots, \psi_{i-1}, \phi, \psi_{i+1}, \ldots, \psi_n) \to \psi$.

\begin{definition}[Named and pasted]
Let $s\in S$ and $\Gamma_s$ be a set of formulas of sort $s$ from 
 ${\mathcal H}_{ \Sigma}(@_z)$. We say that 
 \vspace*{-0.2cm}
 \begin{itemize}
 \item $\Gamma_s$ is {\sf named} if one of its elements is a nominal,
 \item $\Gamma_s$ is {\sf pasted} if it is both {\sf 0-pasted} and {\sf 1-pasted}:
 \begin{itemize}
 \item[{\rm (-)}] $\Gamma_s$ is {\sf 0-pasted} if, for any $t\in S$,  $\sigma\in\Sigma_{s_1\cdots s_n,t}$, $z$ a state symbol of sort $t$, and $\phi$ a formula of sort $s_i$, whenever $@_z^s\phi\in \Gamma_s$ there exists a nominal $j\in {\rm NOM}_{s_i}$ such that $@_z^s\sigma(\ldots, \phi_{i-1},j \wedge \phi,\phi_{i+1},\ldots)\in \Gamma_s$. 
 \item[{\rm (-)}] $\Gamma_s$ is {\sf 1-pasted} if, for any $t\in S$,  $\sigma\in\Sigma_{s_1\cdots s_n,t}$, $z$ a state symbol of sort $t$, and $\phi$ a formula of sort $s_i$, whenever $@_z^s\sigma(\ldots, \phi_{i-1},\phi,\phi_{i+1},\ldots)\in \Gamma_s$ there exists a nominal $j\in {\rm NOM}_{s_i}$ such that $@_z^s\sigma(\ldots, \phi_{i-1},j \wedge \phi,\phi_{i+1},\ldots)\in \Gamma_s$. 
 \end{itemize}
%\item $\Gamma_s$ is {\sf $@$-witnessed} if the following two conditions are satisfied:
%\begin{itemize}
% \item[{\rm (-)}] for $s',t\in S$ , $x\in {\rm SVAR}_t$, $k\in {\rm NOM}_{s'}$ and any formula $\phi$ of sort $s'$,  whenever $@_k^s\exists x\, \phi\in \Gamma_s$ there exists $j\in {\rm NOM}_t$ such that $@_k^s\phi[j/x]\in\Gamma_s$,
% \item[{\rm (-)}] for any $t\in S$ and $x\in {\rm SVAR}_t$ there is $j_s\in {\rm NOM}_t$ such that $@_{j_x}^s x\in \Gamma_s$. 
 %\end{itemize}
 \end{itemize}
 \end{definition}

 \begin{lemma}[Extended Lindenbaum Lemma]\label{lemma:lind}
Let $\Lambda$ be a set of  formulas in the language of ${\mathcal H}_{\Sigma}(@_z)$  and $s\in S$. Then any consistent set $\Gamma_s$ of formulas of sort $s$ from  ${\mathcal H}_{\Sigma}(@_z)+\Lambda$  can be extended to a named, pasted and $@$-maximal consistent set by adding countably many nominals to the language.   
\end{lemma}

\begin{proof}
The proof generalizes to the $S$-sorted setting well-known proofs for the mono-sorted hybrid logic, see \cite[Lemma 7.25]{mod}, \cite[Lemma 3, Lemma 4]{pureax}, \cite[Lemma 3.9]{hyb}. 

For each sort $s\in S$, we add a set of new nominals and enumerate this set. Given a set of formulas $\Gamma_s$, define $\Gamma_s^k$ to be $\Gamma_s \cup \{ k_s\} $, where $k_s$ is the first new nominal of sort $s$ in our enumeration. As showed in \cite{noi2}, $\Gamma_s^k$ is consistent.

%Suppose $\Gamma_s^k$ is not consistent. Then there exists some conjunction of formulas $\theta \in \Gamma_s$ such that $\vds{s} k \to \neg \theta$. We use the $(Gen@)$ rule and the $(K@)$ axiom to prove that \mbox{$\vds{s} @_k^s k \to @_k^s \neg \theta$}. From the $(Ref)$ axiom and the $(MP)$ rule it follows $\vds{s}@_k^s \neg \theta$. Remember that $k$ is a new nominal, so it does not occur in $\theta$ and we use $(Name@)$ rule to get that $\vds{s} \neg \theta\Rightarrow \neg \theta \in \Gamma_s$. But this contradicts the consistency of $\Gamma_s$. Now, we prove the case for the additional $@_{j_x}^s x$ formulas. Suppose $\vds{s} \theta \to \neg @_{j_x}^s x$. We use the $(SelfDual)$ axiom to get $\vds{s} \neg \theta \vee @_{j_x}^s \neg x$. If \mbox{$\vds{s} \neg \theta $}, this contradicts the consistency of $\Gamma_s$. If \mbox{$\vds{s}  @_{j_x}^s \neg x$}, then $\mos{s} @_{j_x}^s \neg x$. Hence, for any model $\mathcal{M}$, any assignment function $g$ and any world $w \in W_s$, we have $\mathcal{M},g, w \mos{s} @_{j_x}^s \neg x$ if and only if $\mathcal{M}, g ,v\mos{s} \neg x$ where $V_s(j_x)=\{v\}$. Then for any model $\mathcal{M}$ and any assignment $g$, $g_s(x) \neq V_s(j_x)$, contradiction.

Now we enumerate  on each sort $s \in S$  all the formulas of the new language obtained by adding the set of new nominals and define $\Gamma^0 := \Gamma_s^k$. Suppose we have defined $\Gamma^m$, where $m \geq 0$. Let $\phi_{m+1}$ be the $m+1-th$ formula of sort $s$ in the previous enumeration. We define $\Gamma^{m+1}$ as follows. If $\Gamma^{m}\cup \{\phi_{m+1}\}$ is inconsistent, then $\Gamma^{m+1} = \Gamma^{m}$. Otherwise:
\begin{itemize}
\item[(i)] $\Gamma^{m+1} = \Gamma^{m} \cup  \{\phi_{m+1}\} $, if $\phi_{m+1}$ is not of the form $@_z\sigma(\ldots, \varphi, \ldots)$ or  $@_x x$, where $\varphi$ a formula of sort $s''$, $x \in {\rm SVAR_{s''}}$ and $z$ is a state symbol.
\item[(ii)] $\Gamma^{m+1} = \Gamma^{m} \cup  \{\phi_{m+1}\} \cup \{@_x  (k \wedge x ) \} $, if $\phi_{m+1}$ is of the form $@_x x $, where $k$ is a new nominal that does not occur in $\Gamma^m$.
\item[(iii)] $\Gamma^{m+1} = \Gamma^{m} \cup  \{\phi_{m+1}\} \cup \{@_x \sigma(\ldots, k \wedge \phi, \ldots)  \} $, if $\phi_{m+1}$ is of the form $@_x \sigma(\ldots, \varphi, \ldots)$ and  $k$ is a new nominal that does not occur in $\Gamma^m$ or $@_x \sigma(\ldots, \varphi, \ldots)$.
%\item[(iv)] $\Gamma^{m+1} = \Gamma^{m} \cup  \{\phi_{m+1}\} \cup \{ @_j \varphi[k/x]\}$, where $\phi_{m+1} $ is of the form $@_j \exists x\varphi(x)$. 
\end{itemize}
In clauses $(ii)$ and $(iii)$, $k$ is the first new nominal in the enumeration that does not occur in $\Gamma^i$ for all $i \leq m$, nor in $@_x \sigma(\ldots, \varphi, \ldots)$.

Let $\Gamma ^+= \bigcup_{n\geq 0} \Gamma^n$. Because $k \in \Gamma^0 \subseteq \Gamma^+$, this set in named, maximal, pasted and $@$-witnessed by construction. We will check if it is consistent for the expansion made in the second, third and fourth items.

Suppose $\Gamma^{m+1} = \Gamma^{m} \cup  \{\phi_{m+1}\} \cup \{@_x(k \wedge x) \} $ is an inconsistent set, where $\phi_{m+1}$ is $@_x x$. Then there is a conjunction of formulas $\chi \in \Gamma^m \cup \{\phi_{m+1}\} $ such that \mbox{$\vds{s} \chi \to \neg @_x(k \wedge x) $} and so \mbox{$\vds{s}@_x(k \wedge x)  \to \neg \chi$.} But $k$ is the first new nominal in the enumeration that does not occur neither in $\Gamma^m$, nor in $@_x x$ and by $Paste0$ rule we get $\vds{s} @_x x \to \neg \chi$. Then $ \vds{s} \chi \to \neg @_x x$, which contradicts the consistency of $\Gamma^m \cup  \{\phi_{m+1}\}$. 

Suppose $\Gamma^{m+1} = \Gamma^{m} \cup  \{\phi_{m+1}\} \cup \{@_x \sigma(\ldots, k \wedge \varphi, \ldots)  \} $ is an inconsistent set, where $\phi_{m+1}$ has the form $@_x \sigma(\ldots, \varphi, \ldots)$. Then there is a conjunction of formulas $\chi \in \Gamma^m \cup \{\phi_{m+1}\} $ such that $\vds{s} \chi \to \neg @_x \sigma(\ldots, k \wedge \varphi, \ldots)$ and so \mbox{$\vds{s} @_x \sigma(\ldots, k \wedge \varphi, \ldots) \to \neg \chi$.} But $k$ is the first new nominal in the enumeration that does not occur neither in $\Gamma^m$, nor in $@_x \sigma(\ldots, \varphi, \ldots)$, therefore, by $Paste1$ rule we get $\vds{s} @_x \sigma(\ldots, \varphi, \ldots) \to \neg \chi$. It follows that $\vds{s} \chi \to \neg @_x \sigma(\ldots, \varphi, \ldots)$, which contradicts the consistency of $\Gamma^m \cup  \{\phi_{m+1}\}$. \end{proof}

%Suppose  $\Gamma^{m+1} = \Gamma^{m} \cup  \{\phi_{m+1}\} \cup \{ @_j \varphi[k/x]\}$ is inconsistent, where $\phi_{m+1}$ is $ @_j \exists x\varphi(x)$. Then there is a conjunction of formulas $\chi \in \Gamma^m \cup \{\phi_{m+1}\}$ such that \mbox{$\vds{s}$ $ \chi \to \neg  @_j \varphi[k/x]$}, where $k$ is the new nominal. By generalization on nominals (Lemma \ref{lem:gennom}) we can prove \mbox{$\vds{s}$} $ \forall y( \chi \to \neg   @_j \varphi[y/x])$, where $y$ is a state variable that does not occur in $\chi \to \neg   @_j \varphi[k/x]$. Using $(Q1)$ axiom, we get \mbox{$\vds{s}$} $ \chi \to   \forall y\neg   @_j \varphi[y/x]$ and by $(SelfDual)$ \mbox{$\vds{s}$} $ \chi \to   \forall y   @_j \neg\varphi[y/x] $. Next, we use $(Barcan@)$ to get \mbox{$\vds{s}$} $ \chi \to  @_j  \forall y \neg \varphi[y/x])$. Because $x$ has no free occurrences in $\varphi[y/x]$, we can prove that $  @_j  \forall y \neg \varphi[y/x]) \leftrightarrow @_j \forall x \neg \varphi$. Therefore, \mbox{$\vds{s}$ $ \chi \to  @_j \forall x \neg \varphi$ }, so \mbox{$\vds{s}$} $ \chi \to  @_j \neg \exists x \varphi $ . Use once again $(SelfDual)$ and we have \mbox{$\vds{s}$} $ \chi \to \neg @_j  \exists x \varphi $. Then $\neg @_j  \exists x \varphi $ $\in \Gamma^m \cup \{\phi_{m+1}\}$, but this contradicts the consistency of $\Gamma^m \cup \{\phi_{m+1}\}$.

\begin{definition}[Named models and natural assignments]\label{def:canonic} For any $s \in S$, let $\Gamma_s$ be a named, pasted and witnessed maximal consistent set and for all state symbols $z$, let $\Delta_z =\{ \varphi \mid @_z^s \varphi \in \Gamma_s \}$. Define $W_s = \{ \Delta_z \mid z$ a state symbol of sort $s $ $\}$. Then, we define $\mathcal{M}=(W, \{ R_{\sigma}\}_{\sigma \in \Sigma}, V)$, the named model generated by the $S$-sorted set $\Gamma =\{ \Gamma_s\}_{s\in S}$, where $R_{\sigma}$ and $V$ are the restriction of the canonical relation and the canonical valuation. We define the natural assignment \mbox{$g_s:{ \rm SVAR}_s \to W_s$} by $g_s(x) = \{ w \in W_s \mid x \in w\} $.
\end{definition}
\begin{lemma}[Existence Lemma]\label{lem:existwit2}
Let $\mathcal{M}=(W, \{ R_{\sigma}\}_{\sigma \in \Sigma}, V)$ be a named model generated by a named and pasted $S$-sorted set $\Gamma$ and let $w$ be a witnessed maximal consistent set. If $\srb (\phi_1, \ldots, \phi_n) \in w$ then there exist witnessed maximal consistent sets $u_i$ such that ${R}_{\srb} wu_1\ldots u_n$ and $\phi_i \in u_i$ for any $i\in [n]$. 
\end{lemma}

\begin{proof}
Let $\srb (\phi_1, \ldots, \phi_n) \in w$, then $@^s_j\srb (\phi_1, \ldots, \phi_n) \in \Gamma_s$, but $\Gamma_s$ is pasted( then $1-pasted$), so there exists $k_1$ a nominal of sort $s_1$ such that $ @^s_j \srb (\phi_1 \wedge k_1, \ldots, \phi_n) \in \Gamma_s$, so $\srb (\phi_1 \wedge k_1, \ldots, \phi_n) \in \Delta_j=w$.
We want to prove that $\Delta_{k_1}, \ldots, \Delta_{k_n}$ are suitable choices for $u_1, \ldots, u_n$.

Let $\psi_1 \in \Delta_{k_1}$. Then $@_{k_1} \psi_1 \in \Gamma_s$ and by agreement property we get $@_{k_1} \psi_1 \in \Delta_{j}$. But $\vds{s} k_1 \wedge \psi_1 \to @_{k_1} \psi_1$ (instance of $(Intro)$ axiom), and by modal reasoning we get $\sigma(@_{k_1} \psi_1, \phi_2, \ldots, \phi_n) \in \Delta_j$. From $(Back)$ axiom, $@_{k_1} \psi_1 \in \Delta_j$ and by using the agreement property, $@_{k_1} \psi_1 \in \Gamma_s$. Hence, $\psi_1 \in \Delta_{k_1}$. 

Now, $\srb(\psi_1, \phi_2, \ldots, \phi_n) \in \Delta_j$, then $@_j \sigma(\psi_1, \phi_2, \ldots, \phi_n) \in \Gamma_s $, but the set is pasted, then exists $k_2$ a nominal of sort $s_2$ such that $@_j\sigma(\psi_1, k_2 \wedge \phi_2, \phi_3, \ldots, \phi_n) \in \Gamma_s$. Then $\sigma(\psi_1, k_2 \wedge \phi_2, \phi_3, \ldots, \phi_n) \in \Delta_{j}$.

Let $\psi_2 \in \Delta_{k_2}$. Then $@_{k_2} \psi_2 \in \Gamma_s$ and by agreement property we get $@_{k_2} \psi_2 \in \Delta_{j}$. But $\vds{s} k_2 \wedge \psi_2 \to @_{k_2} \psi_2$ (instance of $(Intro)$ axiom), and by modal reasoning we get $\sigma(\psi_1, @_{k_2} \psi_2,\phi_3, \ldots, \phi_n) \in \Delta_j$. From $(Back)$ axiom, $@_{k_2} \psi_2 \in \Delta_j$ and by using the agreement property, $@_{k_2} \psi_2 \in \Gamma_s$. Hence, $\psi_2 \in \Delta_{k_2}$. Therefore, by induction, we get that $\psi_i \in \Delta_{k_i}$ for any $i \in [n]$. Then $@_{k_i} \psi_i \in \Gamma_s$ if and only if, by agreement property, $@_{k_i} \psi_i \in \Delta_j$. But $ \sigma( k_1, \ldots,k_n) \in \Delta_j$ and by using $(Bridge)$, it follows that $\sigma( \psi_1, \ldots, \psi_n) \in \Delta_j$. We proved that for any $i \in [n]$, $\psi_i \in \Delta_{k_i}$ we have $\sigma( \psi_1, \ldots, \psi_n) \in \Delta_j$ and by Definition \ref{def:canonic}, it follows that $R_{\sigma}\Delta_j \Delta_{k_1} \ldots \Delta_{k_n}$.
\end{proof}

\begin{lemma}[Truth Lemma]\label{truthlemma2}
Let $\mathcal{M}$ be an $(S, \Sigma)$-model, $g$ an $\mathcal{M}$-assignment function and $w$ a maximal consistent set. For any sort $s\in S$ and any formula $\phi$ of sort $s$, we have:
\begin{center}
$\phi \in w$ if and only if $\mathcal{M}, g, w \mos{s} \phi$.
\end{center}
\end{lemma}
\begin{proof}
We make the proof by structural induction on $\phi$.
\begin{itemize}
\item  $\mathcal{M},g,w \mos{s} a$,where $a\in {\rm PROP}_s\cup {\rm NOM}_s$, if and only if $w\in V_s(a)$ if and only if $a\in w$;
\item $\mathcal{M},g,w \mos{s} x$, where $x\in {\rm SVAR}_s$, if and only if $w=g_s(x)$, if and only if $ x\in w$;
\item $\mathcal{M},g,w \mos{s} \neg \phi$ if and only if $\mathcal{M},g,w \not\mos{s}\phi$ if and only if $\phi \not\in w$ (inductive hypothesis) if and only if $\neg \phi \in w $ (maximal consistent set);
\item $\mathcal{M},g,w \mos{s} \phi \vee \psi$ if and only if $\mathcal{M},g,w \mos{s} \phi$ or $\mathcal{M},g,w \mos{s} \psi$  if and only if $\phi \in w$ or $\psi \in w$ (inductive hypothesis) if and only if $\phi \vee \psi \in w$;
\item let $\tr \in \Sigma_{s_1 \ldots s_n,s}$ and $\phi=\tr (\phi_1, \ldots , \phi_n )$; then $\mathcal{M},g, w \mos{s} \srb (\phi_1, \ldots , \phi_n )$, if and only if for any $i \in [n]$ there exist $u_i \in W_{s_i}$ such that ${R}_{\srb} wu_1\ldots u_n$ and $\mathcal{M},g,u_i  \mos{s_i} \phi_i$ if and only if  for any $i \in [n]$ there exist $u_i \in W_{s_i}$ such that $\phi_i \in u_i$ and  ${R}_{\srb} wu_1\ldots u_n$ (induction hypothesis)  if and only if  $\srb (\phi_1, \ldots , \phi_n ) \in w$ (using Existence Lemma \ref{lem:existwit2}).  

\item $\mathcal{M}, g, w \mos{s} @^s _z \phi$ if and only if $\mathcal{M}, g, \Delta_z \mos{s} @^s _z \phi \in \Delta_z$ (by Lemma \ref{lem:prop}.(3)) if and only if $\phi \in \Delta_z$ (inductive hypothesis) if and only if $@^s_z \phi$ (by $Intro$ axiom together with $z \in \Delta_z$) if and only if $@^s_z \phi \in w$ (by Lemma \ref{lem:prop}.(2)).
\end{itemize}
\end{proof}

\begin{theorem}[Completeness]
Every consistent set of formulas is satisfied.
\end{theorem}
\begin{proof}
Let $\Gamma_s$ be an $s$-sorted set of formulas. By the Extended Lindenbaum Lemma \ref{lemma:lind} we can expand it to a named and pasted set $\Gamma^{+}_s$. By the Truth Lemma \ref{truthlemma2}, the named and natural assignment that $\Gamma^{+}_s$ give rise to satisfy $\Gamma_s$ at $\Gamma^{+}_s$.
\end{proof}

\medskip
\subsection{Example}

Modal logic has traditionally been used for program verification, one of the most remarkable examples being Propositional Dynamic Logic (PDL), which can  represent Hoare Logics. In both Hoare Logics and Dynamic Logic programs are verified using axiomatic semantics, while the state transition system is only semantically defined. For a general discussion we refer to \cite{dynamic}.

Our many-sorted setting allows us to define  both the syntax of a programming  language and its evaluation context in the syntactic layer of our logic, and consequently to define its operational semantics. The change of a configuration after the execution of a program is represented as an implication in our logic, the configuration and the programs being formulas of appropriate sorts.

Our goal is to express operational semantics of languages as axioms in this logic, and to make use of such semantics in program verification. We consider here the SMC Machine described by Plotkin \cite{plotkin}, we derive a Dynamic Logic set of axioms from its proposed transition semantics, and we argue that this set of axioms can be used to derive Hoare-like assertions regarding functional correctness of programs written in the SMC machine language.

The semantics of the \textit{SMC machine} as laid out by Plotkin consists of a set of transition rules defined between configurations of the form $\left\langle S, M, C \right\rangle$, where $S$ is a value stack of intermediate results, $M$ represents the memory, mapping program identifiers to concrete values, and $C$ is the control stack  of commands representing the control flow of the program.

Inspired by the \textit{Propositional Dynamic Logic} (PDL) \cite{dynamic}, we identify a command from the control stack  with a ``program'' from PDL, and use the ``;'' operator from PDL to denote stack composition. We define our formulas to stand for {\em configurations} of the form $config(vs, mem)$
comprising only a value stack and a memory.

Similarly to PDL, we use the modal operator
$ [\_]\_ : CtrlStack \times Config \to Config$
to assert that a configuration formula must hold after executing the commands in the control stack.
The axioms defining the dynamic logic semantics of the SMC machine are then formulas of the form 
$cfg \to [ctrl] cfg'$ 
saying that a configuration satisfying $cfg$ must change to one satisfying $cfg'$ after executing $ctrl$.

\begin{figure}
\small{

\noindent \begin{minipage}[t]{.45\textwidth}

{\bf Syntax}
\vspace*{0.3cm}

$\begin{array}{rcl}
 Nat & ::= & natural \,\, numbers\\
 Var & ::= &  program \,\, variables\\
Bool & ::= & true \, |\,  false\\
AExp & ::= & Nat \, \,|\,\, Var\,\, \\
& & | \,\,AExp\, \te{+}\, AExp\\
BExp & ::= & AExp\, \te{<=}\, AExp\\
Stmt & ::= & x\, \te{:=}\, AExp\\
     & &   |\,\, \te{if}\,\, BExp \\
      & &    \,\,\,\te{then}\,\, Stmt \\
       & &   \,\,\, \te{else}\,\, Stmt\\
        & & |\, \,\te{while}\,\, BExp \,\,\te{do}\,\, Stmt\\
        & & |\,\, \te{skip}\\
        & & |\,\, Stmt \,\te{;}\,Stmt
\end{array}$
\end{minipage}
\begin{minipage}[t]{.45\textwidth}
{\bf Semantics}
\vspace*{0.3cm}

$\begin{array}{rcl}
   Val &::=& Nat\,\, |\,\, Bool\\
ValStack &::=& nil \\
           & &  |\,\, Val\, . \, ValStack\\
     Mem &::= &empty \,\,| \,\, set(Mem, x, n) \\
          & &   |\, get(x,n)\\
CtrlStack &::=& c(AExp)\\
          & &    | \,\, c(BExp) \\ 
           & &   |\,\, c(Stmt)\\
          & &    |\,\, asgn(x)  \\ 
          & &   |\,\, plus \,\,  |\,\, leq   \\ 
          & &   |\,\, Val ?\\
          & &   |\,\, c1 ; c2\\
  Config  & ::= &  config(ValStack, Mem)          

\end{array}$
\end{minipage}
}
\vspace*{3mm}
\caption{\bf Signature}\label{SMC}
\end{figure}

In Figure \ref{SMC}, we introduce the signature of our logic as a context-free grammar (CFG) in a BNF-like form. We make use of the established equivalence between CFGs and algebraic signatures (see, e.g., \cite{HHKR89}), mapping non-terminals to sorts and CFG productions to operation symbols.
Note that, due to non-terminal renamings (e.g., $Exp ::= Int$), it may seem that our syntax relies on sub-sorting.
However, this is done for readability reasons only.
The renaming of non-terminals in syntax can be thought of as syntactic sugar for defining injection functions.
For example, $Exp ::= Int$ can be thought of as $Exp ::= int2Exp(Int)$, and all occurrences of an integer term in a context in which an expression is expected could be wrapped by the $int2Exp$ function.

The sorts $CtrlStack$ and $Config$ correspond to "programs" and "formulas" from PDL, respectively. 
 {Therefore the usual operations of dynamic logic $;$ 
 (composition), $\cup$ (reunion), $^*$ (repetition),  $[\_]\_$ are defined accordingly \cite[Chapter 5]{dynamic}. We depart from PDL with the definition of ``$?$'' (test): in our setting, in order to take a decision, we test the top value of the value stack. Consequently,   the signature of the test operator is $?:Val\to CtrlStack$.}

\medskip

We are ready to define our axioms. For the rest of the paper, whenever $\phi$ is a theorem of sort $s$, i.e. $\vds{s}\phi$, we will simply write $\vdash\phi$, since the sort $s$ can be easily inferred. 

\paragraph{PDL-inspired axioms.}
The first group of axioms is inspired by the axioms of PDL \cite[Chapter 5.5]{dynamic}.
$\pi$, $\pi'$ are formulas of sort ${CtrlStack}$ ("programs"), $\gamma$ is a  formula of sort ${Config}$ (the analogue of "formulas" from PDL), $v$ and $v'$ are variables of sort ${Var}$, $vs$ has the sort ${ValStack}$ and $mem$ has the sort ${Mem}$.\\

$\begin{array}{ll}
(A\cup) & [\pi \cup \pi'] \gamma \leftrightarrow [\pi] \gamma \wedge [\pi'] \gamma \\
(A;) & [\pi;\pi'] \gamma \leftrightarrow [\pi][ \pi'] \gamma \\
(A^*) & [\pi^*] \gamma \leftrightarrow \gamma \wedge [\pi][\pi^*] \gamma\\
(A?) &  config(v \cdot  vs, mem) \to [v ?] config(vs,mem) \\

(A\neg ?) &  config(v \cdot  vs, mem) \to [v' ?] \gamma  \mbox{ where } v \mbox{ and } v' \mbox{ are distinct}.
\end{array}$

\paragraph{SMC-inspired axioms. }
Next, we encode the transition system of the SMC machine as a set of axioms. Apart from the axioms for memory (which are straight-forward), we follow the rules of the SMC machine as closely as allowed by the formalism, using the same notation as in \cite{plotkin}. The sort of each variable can be easily deduced.\\

$\begin{array}{ll}
(CStmt)& c(s1\te{;}s2)\lra c(s1);c(s2)\\ 
(AMem0) & empty \to get(x, 0)\\

(AMem1)  &  set(mem, x, n) \to  get(x, n)\\
(AMem2) & set(set(mem, x, n),y,m)\lra set(set(mem, y,m),x,n) \\ & \mbox{where } x \mbox{ and } y \mbox{ are distinct}\\
(AMem3) & set(set(mem, x, n),x,m)\to set(mem, x,m)\\

(Aint) &   config(vs,mem) \to [c(n)] config(n \cdot vs,mem)   \\ & \mbox{where } n \mbox{ is an integer}\\
(Aid)  &  config(vs, set(mem,x,n)) \to [c(x)] config(n \cdot vs,set(mem,x,n))\\
(Dplus) &    c(a1\, \te{+} \, a2) \lra c(a1) ; c(a2) ; plus \\
(Aplus)  &  config(n2 \cdot n1 \cdot vs,mem) \to [plus] config(n \cdot vs,mem)\\
 &   \mbox{where } n \mbox{ is } n1 + n2 \\

(Dleq)   & c(a1 \,\te{<=}\, a2) \lra c(a2) ; c(a1) ; leq \\

(Aleq)    &config(n1 \cdot n2 \cdot vs,mem) \to  [leq] config(t \cdot vs,mem) \\ 
& \mbox{where } t \mbox{ is the truth valueo of } n1 \leq n2 \\
(Askip) &    \gamma \to  [c(\te{skip})]\gamma \\
(Dasgn) &   c(x\, \te{:=} \, a) \lra c(a) ; asgn(x)\\
(Aasgn)  &  config(n \cdot vs,mem) \to [asgn(x)] config(vs,set(mem, x, n))\\
(Dif)  &  c(\te{if}\,\, b\,\, \te{then}\,\, s1 \,\,\te{else}\,\, s2) \lra c(b) ;
       ( (true \, ?; c(s1))\cup (false \, ?; c(s2)) ) \\
(Dwhile)  &  c(\te{while}\,\, b\,\, \te{do}\,\, s) \lra c(b) ; 
(true ? ; c(s); c(b))^*; false ? \\
\end{array}$

\bigskip

The system $\mathcal{H}_{\Sigma}(@_z)$ presented in this paper can be used to certify executions, but we still cannot perform symbolic verification similarly with the system presented in \cite{noi}.
 
 \medskip

We conclude by a simple example formalizing and stating a formula which can be proven by deduction in our logic. Let $pgm$ be the following  program
\begin{center}

\texttt{i1:= 1; i2:= 2; if i1<=i2 then m:= i1 else m:= i2}
\end{center}
\noindent Note that $pgm$ is a formula of sort ${Stmt}$ in our logic, $m$ is a formula of sort ${Var}$ and $1$ is a formula of sort ${Nat}$. For this formula we have proved in \cite{noi} the following property:

\medskip

\noindent (P$_{pgm}$) $\,\,\vdash  config(vs,mem)\to[c(pgm)]config(vs,mem')$ implies\\
\hspace*{1.1cm} $\vds{Mem}\, mem'\to get(m,1)$
 
\noindent{for any  $mem, mem'$ of sort ${Mem}$ and  $vs$  of sort 
${ValStack}$.
}
\medskip

Which, can be read in plain English as: after executing $pgm$ the value of the program variable $m$ (in memory) will be $1$, and the value stack will be the same as before the execution.

But $\mathcal{H}_{\Sigma}(@_z)$ is an enriched system with the satisfaction operator and we will show that for this system we can prove the following property:

\medskip
\noindent(P') $config(vs, mem) \to [c(pgm)]@_{mem'}~ get(m,1)$
\medskip

In \cite{noi} we have already proved that:

\noindent $\vdash config(vs,mem)\to [c(pgm)]config(vs, set(set(set(mem, i2,2), i1,1),m,1))$

But in order to carry on with the proof of the new property, we need to add a new axiom for the constructor $config$ in order to perform unification:

\medskip
\noindent $(NoConfusion)$ $config(\phi_1, \psi_1)  \wedge config(\phi_2, \psi_2) \to config( \phi_1\wedge \phi_2, \psi_1 \wedge \psi_2)$

\medskip
We refer to \cite{rosu} for a general discussion.

Due to lack of space and in order to ease understanding, from this point on we will use the following notation: 
$mf = set(set(set(mem, i2,2), i1,1),m,1)$
\medskip

{\bf Proof} of (P'):

\noindent(1) $config(vs,mem)\to [c(pgm)]config(vs,mf) $ 

\noindent(2) $config(vs,mem)\to[c(pgm)]config(vs,mem')$

\noindent(3)  $config(vs,mem)\to ([c(pgm)]config(vs, mf) ~\wedge~[c(pgm)]config(vs,mem')) $ \\
 $~$ \hfill PL:(1),(2)

\noindent(4) $([c(pgm)]config(vs, mf) \wedge[c(pgm)]config(vs,mem'))  \to$\\
$\hspace*{1cm}$ $~$ \hfill $ [c(pgm)](config(vs, mf) \wedge config(vs,mem'))$~~ ML

\noindent(5) $(config(vs, mf) \wedge config(vs,mem') ) \to config(vs \wedge vs, mf \wedge mem')$ \\ 
$~\ $ \hfill $(NoConfusion)$

\noindent(6) $config(vs \wedge vs, mf \wedge mem') \to config(vs \wedge vs, @_{ mem'}~ mf)$\\
$~\ $  \hfill $(Intro)$, ML

\noindent(7) $config(vs \wedge vs, @_{ mem'}~ mf) \to @_{ mem'} ~mf$ \hfill $(Back)$

\noindent(8) $[c(pgm)](config(vs \wedge vs, @_{ mem'} ~mf) \to @_{ mem'}~ mf)$ \hfill (UG)

\noindent(9) $[c(pgm)]config(vs \wedge vs, @_{ mem'}~ mf) \to [c(pgm)]@_{ mem'}~ mf ~$ \hfill $~(K_{\sigma}),(MP)$

\noindent(10) $mf \to get(m,1)$ \hfill (AMem2)

\noindent(11) $@_{mem'} ~mf \to @_{mem'}~ get(m,1)$  \hfill ML:(10)

\noindent(12) $[c(pgm)]@_{mem'} ~mf \to [c(pgm)]@_{mem'}~ get(m,1)$ \hfill (UG),$(K_{\sigma}),(MP)$

\noindent(13) $[c(pgm)](config(vs, mf) \wedge config(vs,mem') ) \to $\\
$~$ \hfill $[c(pgm)] config(vs \wedge vs, mf \wedge mem')~~$ (UG):(5), $(K_{\sigma}),(MP)$

\noindent(14) $[c(pgm)]config(vs \wedge vs, mf \wedge mem') \to [c(pgm)]config(vs \wedge vs, @_{ mem'}~ mf)$ \\
$~$ \hfill (UG):(6), $(K_{\sigma}),(MP)$

\noindent(15) $config(vs, mem) \to [c(pgm)]@_{mem'}~ get(m,1)$ PL:(3),(4),(13),(14),(9),(12)

\section{The many-sorted hybrid  modal logic ${\mathcal H}_{ \Sigma}(\forall)$}\label{sec:hyb}

The hybridization of our many-sorted modal logic is developed  using a combination of ideas and techniques from \cite{hand,pureax,hyb,mod,goranko,goranko2}, but for this section we drew our inspiration mainly from \cite{hyb}. %We refer to \cite{noi2} for some similar proofs of the results presented in this section.

Hybrid logic is defined on top of modal logic by adding {\em nominals}, {\em states variables} and specific binders. This is a first step towards employing the procedure of hybridization on top of the many-sorted polyadic modal logic. The main idea was to define a general logical system that is powerful enough to represent both the programs and their semantics in an uniform way.

Once again , the sorts will be denoted by $s$, $t$, $\ldots$ and by ${\rm PROP}=\{{\rm PROP}_s\}_{s\in S}$, ${\rm NOM}=\{{\rm NOM}_s\}_{s\in S}$ and ${\rm SVAR}=\{{\rm SVAR}_s\}_{s\in S}$ we will denote the same countable $S$-sorted sets presented in Section \ref{sec:basichyb}.
%Therefore, we can define our logic also on top of the basic hybrid logic previously presented in Section \ref{sec:basichyb} by adding the needed binders.
% Nominals  allow us to directly refer the worlds (states) of a model, since they are evaluated in singletons in any model. However, a nominal may refer different worlds in different models.
%The sorts will be denoted by $s$, $t$, $\ldots$ and by ${\rm PROP}=\{{\rm PROP}_s\}_{s\in S}$, ${\rm NOM}=\{{\rm NOM}_s\}_{s\in S}$ and ${\rm SVAR}=\{{\rm SVAR}_s\}_{s\in S}$ we will denote some countable $S$-sorted sets. The elements of ${\rm PROP}$ are ordinary propositional variables and they will be denoted $p$, $q$,$\ldots$; the elements of ${\rm NOM}$ are called {\em nominals} and they will be denoted by $j$, $k$, $\ldots$; the elements of ${\rm SVAR}$ are called {\em state variables} and they are denoted $x$, $y$, $\ldots$.  We shall assume  that for any distinct sorts $s\neq t\in S$, the corresponding sets of propositional variables, nominals and state variables are distinct. A {\em state symbol} is a nominal or a state variable. 

\begin{definition}[${\mathcal H}_{ \Sigma}(\forall)$ formulas] For any sort $s\in S$ we define the formulas of sort $s$:
\begin{center}
$\phi_s :=  p\mid j\mid y_s\mid \neg \phi_s \mid\phi_s \vee \phi_s \mid \sigma(\phi_{s_1}, \ldots, \phi_{s_n})_s
  \mid  \forall x_t\, \phi_s$
\end{center}
\noindent Here, $p\in {\rm PROP}_s$, $j\in {\rm NOM}_s$, $t\in S$, $x\in {\rm SVAR}_t$, $y\in {\rm SVAR}_s$ and  $\sigma\in \Sigma_{s_1\cdots s_n,s}$.

 We also define the {dual binder} $\exists$. For any $s,t\in S$, if $\phi$ is a formula of sort  $s$ and $x$ is a state variable of sort $t$, then
\begin{center}
  $\exists x\, \phi := \neg\forall x\, \neg\phi$ is a formula of sort $s$. 
 \end{center} 

The notions of {\sf free state variables} and {\sf bound state variables} are defined as usual.  
\end{definition}

%In order to define the semantics for ${\mathcal H}_{ \Sigma}(@_z,\forall)$ more is needed. 
Given a model ${\mathcal M}=(W, (R_\sigma)_{\sigma\in\Sigma}, V)$, then \mbox{$g : {\rm SVAR} \rightarrow W$} is an {\em assignment} is an $S$-sorted function. If $g$ and $g'$ are assignment functions $s\in S$ and \mbox{$x\in \mbox{SVAR}_s$} then we say that $g'$ is an {\em $x$-variant} of $g$ (and we write $g'\stackrel{x}{\sim} g$)  if $g_t=g'_t$ for 
$t\neq s\in S$ and  $g_s(y)=g'_s(y)$ for any \mbox{$y\in \mbox{SVAR}_s$,} $y \neq x$.

The satisfaction relation is defined similar with the one in ${\mathcal K}_{\Sigma}$, but we only need to add the definition for binders:
%\begin{definition}[The satisfaction relation in  ${\mathcal H}_{ \Sigma}(\forall)$] In the sequel we have the model ${\mathcal M}=(W, (R_\sigma)_{\sigma\in\Sigma}, V)$ and $g:{\rm SVAR}\to W$ an $S$-sorted assignment. The satisfaction relation is defined similar with the one in Section \ref{sec1}
%\begin{itemize}
%\item $\mathcal{M},g,w \mos{s} a$, if and only if $w\in V_s(a)$, where $a\in {\rm PROP_s}\cup {\rm NOM_s}$,
%\item $\mathcal{M},g,w \mos{s} x$, if and only if $w=g_s(x)$, where $x\in {\rm SVAR}_s$,
%\item $\mathcal{M},g,w \mos{s} \neg \phi$, if and only if $\mathcal{M},g,w \mosn{s}\phi$
%\item $\mathcal{M},g,w \mos{s} \phi \vee \psi$, if and only if $\mathcal{M},g,w \mos{s} \phi$ or $\mathcal{M},g,w \mos{s} \psi$ 
%\item if $\sigma\in\Sigma_{s_1\ldots s_n,S}$ then $\mathcal{M},g,w \mos{s} \sigma(\phi_1, \ldots , \phi_n )$, if and only if there is \\$(w_1,\ldots,w_n) \in W_{s_1}\times\cdots\times W_{s_n}$ such that  $R_{\sigma} ww_1\ldots w_n$ and $\mathcal{M},g,w_i  \mos{s_i} \phi_i$ for any $i \in [n]$,
%%\item $\mathcal{M},g,w \mos{s} @_k^s\phi$ if and only if $\mathcal{M},g,u\mos{t} \phi$ where $k\in {\rm NOM}_t\cup {\nom}_t$,  $\phi$ has the sort $t$ and $V^{\nom}_t(k)=\{u\}$,
% \item 
 \begin{center}
 $\mathcal{M},g,w \mos{s} \forall x\,\phi$, if and only if  $\mathcal{M},g',w \mos{s} \phi$ for all $g'\stackrel{x}{\sim} g$.
 \end{center}

Consequently, 
$\mathcal{M},g,w \mos{s} \exists x\, \phi$, if and only if $\exists g'( g' \stackrel{x}{\sim} g \ and \  \mathcal{M},g',w \mos{s} \phi)$.

 In order to define the axioms of our system, one more definition is needed.

We assume  $\#_s$ be a new propositional variable of sort $s$ and  
we inductively define $NC=\{NC_s\}_s$ by
\begin{itemize}
\item $\#_s,\top_s\in NC_s$ for any $s\in S$
\item if $\sigma\in \Sigma_{s_1\cdots s_n,s}$  and $\eta_i\in NC_{s_i}$ for any $i\in [n]$ then $\sigma(\eta_1,\ldots, \eta_n)\in NC_s$.
\end{itemize}

\noindent We further define $NomC=\{NomC_s\}_{s\in S}$ such that $\eta\in NomC_s$ iff $\eta\in NC_s$ and $|\{\#_s\mid s\in S, \#_s\,\in\,\, \eta\}|=1$. If  $\eta\in NomC_{s}$ then $\eta^{\mb}$ is its dual and  $\eta(\vp) \,:=\, \eta [\vp /\#_{s'}]$.

\begin{remark} If $\eta\in NomC_{s}$ and $\vp\in Form_{s'}$ then 
$\mathcal{M},g,w \mos{s}\eta(\vp)  \ \mbox{iff} \ \mathcal{M},h,w' \mos{s'} \vp$ for some $w'$ in the submodel generated by $\mathcal X$ where ${\mathcal X}_s=\{w\}$ and ${\mathcal X}_{t}=\emptyset$ for $t\neq s$. Dually, $\mathcal{M},g,w \mos{s}\eta^{\mb}(\vp)  \ \mbox{iff} \ \mathcal{M},h,w' \mos{s'} \vp$ for any $w'$ in the submodel generated by $\mathcal X$. 
\end{remark}

The deductive system is presented in Figure \ref{fig:k}.

\begin{figure}[h]
\centering

%{\bf The system   ${\mathcal H}_{ \Sigma}$}
%\begin{itemize}
%\item For any $s\in S$, if $\phi$ is a formula of sort $s$ which is a theorem in propositional logic, then $\phi$ is an axiom. 
%\item Axiom schemes: for any $\sigma\in \Sigma_{s_1\cdots s_n,s}$ and for any formulas $\phi_1,\ldots, \phi_n,\phi,\chi$, $\psi$ of appropriate sorts, the following formulas are axioms:
%
%
%$\begin{array}{rl}
%(K_\sigma) & \sigma^{\mb}(\ldots,\phi_{i-1},\phi\rightarrow\chi,\phi_{i+1}, \ldots)\to( \sigma^{\mb}(\ldots ,\phi_{i-1}, \phi, \phi_{i+1},\ldots) \to \sigma^{\mb}(\ldots ,\phi_{i-1}, \chi, \phi_{i+1},\ldots))\\
%(Dual_\sigma)& \sigma (\phi_1,\ldots ,\phi_n )\leftrightarrow \neg \sigma^{\mb} (\neg \phi_1,\ldots ,\neg \phi_n )
%\end{array}$
%\item Deduction rules: \\
%\begin{tabular}{rl}
%$(MP)$ & if $\vds{s}\phi$ and $\vds{s}\phi\to \psi$ then 
%$\vds{s}\psi$\\
%$(UG)$ &  if $\vds{s_i}{\phi}$ then $\vds{s}\sigma^{\mb} (\phi_1, .. ,\phi, ..\phi_n)$
%\end{tabular}
%\end{itemize}

{\bf The system   ${\mathcal H}_{ \Sigma}(\forall)$}

\begin{itemize}
\item The axioms and the deduction rules of ${\mathcal K}_{\Sigma}$
\item Axiom schemes: for any $\sigma\in \Sigma_{s_1\cdots s_n,s}$ and for any formulas $\phi_1,\ldots, \phi_n,\phi,\psi$ of appropriate sorts, the following formulas are axioms:

$\begin{array}{rl}
(Q1) & \forall x\,(\phi \to \psi) \to (\phi \to \forall x\, \psi)
  \mbox{ where $\phi$ contains no free occurrences of x}\\
(Q2) & \forall x\,\phi \to \phi[y\slash x] 
  \mbox{ where $y$ is substitutable for $x$ in $\phi$}\\
(Name) & \exists x \, x \\ 
(Barcan) & \forall x \,\sigma^{\mb}(\phi_1, \ldots, \phi_n) \to \sigma^{\mb}(\phi_1, \ldots,\forall x\phi_{i},\ldots, \phi_n)\\
(Nom) & \forall x\, [\eta(x\wedge\phi)\to \theta^{\mb}(x\to\phi)],\\  & $
 for any $s\in S$, $\eta$ and $\theta\in {NomC}_s$, $x\in {\rm SVAR}_{s'} 
\end{array}$
\item Deduction rules: \\
\begin{tabular}{rl}
$(Gen)$ & if $\vds{s}\phi$ then $\vds{s}\forall x\phi$, where $\phi\in Form_s$ and $x\in {\rm SVAR}_t$ for some $t\in S$.
\end{tabular}
\end{itemize} 

\caption{$(S,\Sigma)$ hybrid logic}\label{fig:k}
\end{figure}
\medskip

In \cite{noi3} we have proved the soundness and completeness of the $\mathcal{H}_{\Sigma}(\forall)$ system.
\section{The many-sorted hybrid  modal logic ${\mathcal H}_{ \Sigma}(@_z,\forall)$}\label{sec2}

In \cite{noi2}, given a concrete language with a concrete SMC-inspired
operational semantics, we showed how to define a corresponding (sound and
complete) logical system and we also proved (rather general) results that allow
us to perform Hoare-style verification. Our approach was to define the weakest
system that allowed us to reach our goals. For that, we needed to define the satisfaction operator only on nominals. 

Furthermore, in \cite{noi3}, in order to establish the connection with Matching logic, we have introduced the $\mathcal{H}_{\Sigma}(@_z,\forall)$ system which allows the satisfaction operators $@_z$ to also range over state variables, not just over nominals.

        Therefore, let $(S,\Sigma)$ be a many-sorted signature. As already announced, in this section we extend the system $\mathcal{H}_{\Sigma}(\forall)$ previously defined by adding the satisfaction operators $@_z^s$ where $s\in S$ and $z$ is a {\em state symbol}, i.e. a nominal or a state variable.
        
         The formulas of  ${\mathcal H}_{ \Sigma}(@_z,\forall)$ are defined as follows:
\begin{center}
$\phi_s :=  p\mid j\mid y_s\mid \neg \phi_s \mid\phi_s \vee \phi_s \mid \sigma(\phi_{s_1}, \ldots, \phi_{s_n})_s
  \mid  \forall x_t\, \phi_s\mid @_z^s\psi_t$
\end{center}

Here, $p\in {\rm PROP}_s$, $j\in {\rm NOM}_s$, $t\in S$, $x\in {\rm SVAR}_t$, $y\in {\rm SVAR}_s$, $\sigma\in \Sigma_{s_1\cdots s_n,s}$, $z$ is a state symbol of sort $t$ and 
$\psi$ is a formula of sort $t$.

The satisfaction relation is defined similar with the one in ${\mathcal H}_{ \Sigma}(\forall)$, but we only need to add the definition for $@_z$: $\mathcal{M},g,w \mos{s} @_z^s\phi$ if and only if $\mathcal{M},g,Den_g(z)\mos{t} \phi$ where $z$ is a state symbol of sort $t$ and $\phi$ is a formula of the same sort $t$. Here, $Den_g(z)$ is the denotation of the state symbol $z$ of sort $s$ in a model $\mathcal{M}$ with an assignment function $g$, where $Den_g(z)=V_s(z)$ if $z$ is a nominal, and $Den_g(z)=g_s(z)$ if $z$ is a state variable.

\begin{figure}[!h]
\centering

{\bf The system} ${\mathcal H}_{\Sigma}(@_z,\forall)$
\begin{itemize}
\item The axioms and the deduction rules of ${\mathcal K}_{\Sigma}$

\item Axiom schemes: any formula of the following form is an axiom, where $s,s',t$ are sorts,  $\sigma\in\Sigma_{s_1\cdots s_n,s}$,  $\phi,\psi, \phi_1,\ldots,\phi_n$ are formulas (when necessary, their sort is marked as a subscript), $x$ is state variable and $y$, $z$ are state symbols:
\end{itemize}

\begin{flushleft}
$\begin{array}{rl}
(K@) & @_z^{s} (\phi_t \to \psi_t) \to (@_z^s \phi \to @_z^s \psi) \\
(SelfDual) & @^s_z \phi_t \leftrightarrow \neg @_z^s \neg \phi_t \\
(Intro)  & z \to (\phi_s \leftrightarrow @_z^s \phi_s)\\
(Agree) &  @_y^{t}@_z^{t'} \phi_s \leftrightarrow @^t_z \phi_s\\
(Ref) & @_z^sz_t \\
(Back) & \sigma(\ldots,\phi_{i-1}, @_z^{s_i} {\psi}_t,\phi_{i+1},\ldots)_s\to @_z^s {\psi}_t 
\end{array}$\\
\end{flushleft}
\medskip

$\begin{array}{rl}
(Q1) & \forall x\,(\phi \to \psi) \to (\phi \to \forall x\, \psi)
  \mbox{ where $\phi$ contains no free occurrences of x}\\
(Q2) & \forall x\,\phi \to \phi[y\slash x] 
  \mbox{ where $y$ is substitutable for $x$ in $\phi$}\\
(Name) & \exists x \, x \\ 
(Barcan) & \forall x \,\sigma^{\mb}(\phi_1, \ldots, \phi_n) \to \sigma^{\mb}(\phi_1, \ldots,\forall x\phi_{i},\ldots, \phi_n)\\
(Barcan@) & \forall x \, @_z\phi \to @_z \forall x\,\phi, \mbox{where } x \neq z\\
 (Nom\, x) & @_z x\wedge @_y x \to @_z y 
\end{array}$

\medskip
\begin{itemize}

\item Deduction rules:
\end{itemize}
\begin{tabular}{rl}
$(BroadcastS)$ & if $\vds{s}@_z^s\phi_t$ then $\vds{s'}@_z^{s'}\phi_t$ \\
 $(Gen@)$& if $\vds{s'} \phi$ then $\vds{s} @_z \phi$, where $z $ and $\phi$ have the same sort $s'$\\
%$(Name@)$ & if $\vds{s} @_j \phi$ then $\vds{s'} \phi$, where $j$ does not occur in $\phi$\\
$(Paste0)$ & if $\vds{s} @^s_z (y \wedge \phi) \to \psi$ then $\vds{s} @_z \phi \to \psi$\\ & where $z$ is distinct from $y$ that does not occur in $\phi$ or $\psi$\\
$(Paste1)$ & if $\vds{s} @^s_z \sigma(\ldots, y \wedge \phi,\ldots)  \to \psi$ then $\vds{s} @^s_z \sigma(\ldots, \phi, \ldots) \to \psi$\\ & where $z$ is distinct from $y$ that does not occur in $\phi$ or $\psi$\\
$(Gen)$ & if $\vds{s}\phi$ then $\vds{s}\forall x\phi$,\\ &  where $\phi\in Form_s$ and $x\in {\rm SVAR}_t$ for some $t\in S$.
\end{tabular}

\caption{$(S,\Sigma)$ hybrid logic}\label{fig:forall@}
\end{figure}

The deductive system is presented in Figure \ref{fig:forall@}.

We have proved the soundness and completeness of the $\mathcal{H}_{\Sigma}(@_z, \forall)$-system in \cite{noi3}.

\section{Standard Translation}\label{st}

Next, we will talk about the relationship between modal and classical logic. We first specify our correspondence language, more precisely, the language we will translate our modal formulas to.

Recall that in our many-sorted polyadic modal logic we have  $\tau = (S, \Sigma)$ a many-sorted signature, where the sorts are denoted by $s$, $t$, $\ldots$ and by ${\rm PROP}=\{{\rm PROP}_s\}_{s\in S}$, ${\rm NOM}=\{{\rm NOM}_s\}_{s\in S}$ and ${\rm SVAR}=\{{\rm SVAR}_s\}_{s\in S}$ the well known $S$-sorted sets. 

We introduce the notation $ar(\sigma)$ which denotes not just the arity of the many-sorted modal operator $\sigma$, but also the sort of the arguments, where $ar(\sigma) = <s_1 \ldots s_n, s>$. 

Let us take a look at an $(S, \Sigma)$-model $\mathcal{M}= (W, \{ R_{\sigma}\}_{\sigma \in \Sigma}, V)$ which is a relational structure where $W$ can be seen as a domain of quantification, each $R_{\sigma}$ a relation over this domain, and $V_s(p)$ is a unary relation for each $p \in {\rm PROP}_s$. On the other hand, if we talk about this model $\mathcal{M}= (W, \{ R_{\sigma}\}_{\sigma \in \Sigma}, V)$ using first-order logic we will make use of a first-order language with a relation symbol $R_{\sigma}$ for each $\sigma \in \Sigma$, and a unary relation symbol (predicate) $P_p$ for every $p \in {\rm PROP}_s$.

For the correspondence language in First-Order Logic (FOL) we will define:
\begin{center}
 $\mathcal{L}_{\tau} :=\mathcal{L}_{\tau} ( \rm{PROP}, \rm{NOM}, \rm{SVAR} ) :=\{ = \} \cup \{ P_p \mid p \in PROP_s\} _{s \in S} \cup \{ R_{\sigma} \mid \sigma \in \Sigma\}$.
\end{center}

Therefore, we consider  $\mathcal{L}_{\tau}$ the first-order language with equality which has unary predicates $P_p$ corresponding to the propositional letters $p \in PROP_s$ where $ar(P_p)= <s>$ if and only if $p \in PROP_s$. We add the $(n+1)$-ary relation symbol $R_{\sigma}$ for each $n$-ary many-sorted modal operator $\sigma$, and we consider that $ar(R_{\sigma}) = <ss_1 \ldots s_n >$ if and only if $ar(\sigma) = <s_1 \ldots s_n, s>$. 

Recall that a model in many-sorted polyadic modal logic is defined by $\mathcal{M}=( W, R_{\sigma}, V)$ where $V : PROP \rightarrow W$. For each model in our logic we define the corresponding one by $\mathcal{M} =(W, R_{\sigma}, P_p)$. We use the same modal relation $R_{\sigma}$ to interpret the relation symbol $R_{\sigma}$ in FOL, and the set $V_s(p)$ to interpret the unary predicates $P_p$. As emphasized in \cite{mod}, there is no mathematical distinction between modal and first-order models; because both modal and first-order models are relational structures. Given the construction of out logic on top of modal logic, we can also transfer this feature when talking about the relation between our logic and FOL. Moreover, we use the $S$-sorted set $\rm{VFOL}= \{ \rm {VFOL} _s\}_{s\in S}$ for the set of first-order variables.

\begin{definition}
Let $x$ be a first-order variable. The standard translation $ST_x$ taking modal formulas to first-order formulas in $\mathcal{M}_{\tau}$ is defined as follows:

\begin{itemize}
\item $ST_x(p)=P_p(x)$, where $p \in PROP_s$
\item $ST_x(y) =(x=y)$, where $y \in SVAR_s$
\item $ST_x(j) = (x=c_j)$, where $j \in NOM_s$
\item $ST_x( \sigma ( \phi_1, \ldots, \phi_n)) = \exists y_1 \ldots \exists y_n (R_{\sigma} x y1 \ldots y_n \wedge ST_{y_1}(\phi_1) \wedge \ldots ST_{y_n}(\phi_n)  )$, where $y_1, \ldots, y_n$ are fresh variables , that is, variables that have not been used so far in the translation.
\item $ST_x (@_y^s \phi) = ST_y(\phi)$
\item $ST_x(\exists y \phi ) = \exists y ST_x(\phi)$
\end{itemize}
\end{definition}

That is, the standard translation maps proposition symbols to unary predicates (that is $P_p(x)$ is true when $p$ holds in world $x$), commutes
with booleans, and handles $\sigma$ by explicit first-order quantification over $R_{\sigma}$-accessible points. The variables $y_1, \ldots, y_n$ that are used in the clauses for $\sigma$ are chosen to be
any new variables, ones that has not been used so far in the translation. Please notice that we are using the same set of symbols for state variables and first-order variables. Moreover, for each nominal $j \in NOM_s$, we introduce a corresponding constant $c_j$ in the first-order language in order to translate the nominals into. Also, the satisfaction operators are translated
by substituting the relevant first-order constant for the free-variable $x$. Note that this translation
returns first-order formulas with at most one free variable $x$, not exactly one. This is because
a constant may be substituted for the free occurrence of $x$. For example, the hybrid formula $@^s_j j$
translates into the first-order sentence $j=j$.

The truth of a formula of $\mathcal{L_{\tau}}$ in a structure $\mathcal{M}$, relative to an assignment function $g : SVAR \rightarrow W $ is given in the classical way. We can write $\mathcal{M} \mfol ST_x(\phi)[x \leftarrow w]$ which means that the first-order formula $ST_x(\phi)$ is satisfied in the usual sense of first-order
logic in the model $\mathcal{M}$ when $w$ is assigned to the free variable $x$. By assigning a value to the free variable , which gives the internal perspective representative for modal logic, we can evaluate a formula inside a model at a certain point.

%This free variable is what allows the internal
%perspective, typical of modal logic, to be mirrored in a classical language: assigning a value to
%this variable is analogous to evaluating a modal formula inside a model at a certain point.

\begin{proposition}[Local and Global Correspondence on Models] Let $(S, \Sigma)$ be a many-sorted signature and $\phi$ a formula of sort $s \in S$.

\begin{itemize}
\item[1)] For all $(S, \Sigma)$-models $\mathcal{M}$ and all states $w$ of $ \mathcal{M}$ :\\ $\mathcal{M}, w \mos{s} \phi $ if and only if $\mathcal{M} \mfol ST_x(\phi)[x \leftarrow w ]$
\item[2)] For all $(S, \Sigma)$-models $\mathcal{M}$: \\$\mathcal{M} \mos{s} \phi $ if and only if $\mathcal{M} \mfol \forall x ST_x (\phi)$.
\end{itemize}
\end{proposition}

\begin{proof}
1)By structural induction over $\phi$.

Let $\mathcal{M} \mfol ST_x(p)[x \leftarrow w]$ if and only if $\mathcal{M} \mfol P_p(x)[x \leftarrow w]$ if and only if $w \in P_p$ if and only if $w \in V_s(p) $ if and only if $\mathcal{M}, g, w \mos{s} p$ for any assignment function $g$ if and only if $\mathcal{M}, w \mos{s} p$.

Let $\mathcal{M} \mfol ST_x(y)[x \leftarrow w]$ if and only if $\mathcal{M} \mfol (x=y)[x \leftarrow w]$ for any $g'$ where $g'(y)=w$ and $g'(z') =g(z')$ for any $z'\neq z$ state variables of sort $t$ and $g_s(z') =g'_s(z')$ for any $s\neq t \in S$ if and only if $\mathcal{M}, g, w \mos{t} y$ for any $g$ if and only if $\mathcal{M}, w \mos{t} y$.

Let $\mathcal{M} \mfol ST_x(j)[x \leftarrow w]$ if and only if $\mathcal{M} \mfol (x=c_j)[x \leftarrow w]$ if and only if $w=c_j$ for any assignment function $g$ if and only if $ w \in V_s(j)$ for any $g$ if and only if $\mathcal{M}, g,w \mos{s} j$ for any $g$ if and only if $\mathcal{M}, w \mos{s} j$.

Let $\tr \in \Sigma_{s_1 \ldots s_n,s}$. Then $\mathcal{M} \mfol ST_x(\sigma( \phi_1, \ldots, \phi_n))[x \leftarrow w]$ if and only if $\mathcal{M} \mfol  \exists y_1 \ldots\exists y_n(R_{\sigma}x y_1 \ldots y_n \wedge ST_{y_1}(\phi_1) \wedge \ldots \wedge ST_{y_n}(\phi_n))[x \leftarrow w ]$ $  \ \ \ \ \ \ \ \ \ \ \ \ \ \ \ \ \ \  \ \ \ \ \ \ \ \ \ \ \ \ \ \ \ \ \ \ $ if and only if there exists $(u_1, \ldots, u_n) \in W_{s_1}\times \cdots \times W_{s_n}$ such that \\
$\mathcal{M} ~\mfol (R_{\sigma}x y_1 \ldots y_n \wedge ST_{y_1}(\phi_1) \wedge \ldots \wedge ST_{y_n}(\phi_n))~[x \leftarrow w, y_1 \leftarrow u_1, \ldots, y_n \leftarrow u_n]$ \\if and only if there exists $(u_1, \ldots, u_n) \in W_{s_1}\times \cdots \times W_{s_n}$ such that $R_{\sigma}w u_1 \ldots u_n$ and $\mathcal{M} \mfol ST_{y_i}(\phi_i)[y_i \leftarrow u_i]$ for any $i \in [n]$ if and only if there exists $(u_1, \ldots, u_n) \in W_{s_1}\times \cdots \times W_{s_n}$ such that $R_{\sigma}w u_1 \ldots u_n$ and $\mathcal{M}, u_i \mos{s_i} \phi_i$ for any $ i \in [n]$(induction hypothesis) if and only if there exists $(u_1, \ldots, u_n) \in W_{s_1}\times \cdots \times W_{s_n}$ such that $R_{\sigma}w u_1 \ldots u_n$ and $\mathcal{M},g, u_i \mos{s_i} \phi_i$ for any $ i \in [n]$ and any $g$ if and only if $\mathcal{M}, g, w \mos{s} \sigma( \phi_1, \ldots, \phi_n)$ for any $g$ if and only if $\mathcal{M}, w \mos{s} \sigma( \phi_1, \ldots, \phi_n)$.

Let $z$ be a state variable of sort $t \in S$. Then $\mathcal{M}, w \mos{s} @_z^s \phi_t$ if and only if $\mathcal{M},g, w \mos{s} @_z^s \phi_t$ for any $g$ if and only if $\mathcal{M},g',u \mos{t}  \phi_t$ for any $g'$ where $g_t'(z)=u$ and $g_t'(z') =g_t(z')$ for any $z'\neq z$ state variables of sort $t$ and $g_s(z') =g'_s(z')$ for any $s\neq t \in S$ if and only if $\mathcal{M}, u \mos{t} \phi_t$ for any $u$ if and only if $\mathcal{M} \mfol ST_z(\phi_t) [z \leftarrow w]$ for any $u$ if and only if $\mathcal{M} \mfol ST_z(\phi_t)$ if and only if $\mathcal{M} \mfol ST_z(\phi_t)[x \leftarrow w]$ if and only if $\mathcal{M} \mfol ST_x(@_z^s \phi_t)[x \leftarrow w]$.

Let $j$ be a nominal of sort $t$. Then $\mathcal{M}, w \mos{s} @_j^s \phi_t$ if and only if $\mathcal{M},g, w \mos{s} @_j^s \phi_t$ for any $g$ if and only if $\mathcal{M},g,u \mos{t}  \phi_t$ for any $g$ and $u \in V_t(j)$ if and only if $\mathcal{M}, u \mos{t} \phi_t$ where $u \in V_t(j)$ if and only if $\mathcal{M} \mfol ST_j(\phi_t) $ if and only if $\mathcal{M} \mfol ST_j(\phi_t)[x \leftarrow w]$ if and only if $\mathcal{M} \mfol ST_x(@_j^s \phi_t)[x \leftarrow w]$.

Let $\mathcal{M} \mfol ST_x(\forall y \phi)[x \leftarrow w] $ if and only if $\mathcal{M} \mfol (\forall y ~ST_z(\phi))[x \leftarrow w]$ if and only if for any $a \in W_t$, $\mathcal{M} \mfol ST_x( \phi) [x \leftarrow w, y \leftarrow a]$

2) Let $x \in VFOL_s$. Then $\mathcal{M} \mfol \forall x ST_x(phi)$ if and only if for any $w \in W_s$, $\mathcal{M} \mfol ST_x(\phi)[x \leftarrow w]$ if and only if for any $w \in W_s$, $\mathcal{M}, w \mos{s} \phi $( use item 1) of this proposition) if and only if $\mathcal{M} \mos{s} \phi $.
\end{proof}

Thus the standard translation gives us a bridge between many-sorted modal logic and classical logic.

\section{Conclusions}

Improving over previous work ~\cite{noi,noi2, noi3}, this paper makes the following contributions:
(1)
We study the $@$-only fragment of the more general hybrid modal logic proposed in \cite{noi2,noi3},
 and provide a sound and complete deduction system for it. This logic is important as it is weaker 
 than the full hybrid modal logic and thus it is expected to have better computational properties.
 Nevertheless, although weaker, we show it can be used to axiomatically express
 operational semantics and to derive proofs for statements concerning program executions.
(2) We provide a standard translation from full hybrid modal logic to first-order logic and prove
that it induces both local and global correspondence on models.

\paragraph{Future Work}

Although the use of quantifiers (particularly existentials) makes for easier to write 
and express statements about programs, the $@$ operator can suplement the need for quantification in many cases.
Exploring the limits of this capacity seems like an interesting path to follow.

The promise of giving up quantification in favor of just $@$ is that we sacrifice expresiveness for
better computational properties. We would like to find out if that indeed is the case, by 
investigating decidability results for the $@$-only fragment of the logic.

%\section{Bibliography}
%
%We request that you use
%\href{http://eptcs.web.cse.unsw.edu.au/eptcs.bst}
%{\tt $\backslash$bibliographystyle$\{$eptcs$\}$}
%\cite{bibliographystylewebpage}, or one of its variants
%\href{http://eptcs.web.cse.unsw.edu.au/eptcsalpha.bst}{eptcsalpha},
%\href{http://eptcs.web.cse.unsw.edu.au/eptcsini.bst}{eptcsini} or
%\href{http://eptcs.web.cse.unsw.edu.au/eptcsalphaini.bst}{eptcsalphaini}
%\cite{bibliographystylewebpage}. Compared to the original {\LaTeX}
%{\tt $\backslash$biblio\-graphystyle$\{$plain$\}$},
%it ignores the field {\tt month}, and uses the extra
%bibtex fields {\tt eid}, {\tt doi}, {\tt ee} and {\tt url}.
%The first is for electronic identifiers (typically the number $n$
%indicating the $n^{\rm th}$ paper in an issue) of papers in electronic
%journals that do not use page numbers. The other three are to refer,
%with life links, to electronic incarnations of the paper.

\nocite{*}
\bibliographystyle{elsarticle-num}

\begin{thebibliography}{8}

  \bibitem{hand}
Areces, C., ten Cate, B.: \emph{Hybrid Logics}. In: Handbook of Modal Logic, P. Blackburn et al. (Editors) 3, pp. 822--868 (2007).

%\bibitem{blsel}
%P. Blackburn, J. Seligman, Hybrid Languages, Journal of Logic, Language and Information, (4):251-272 (1995).

 \bibitem{pureax}
Blackburn, P., ten Cate, B.: \emph{Pure Extensions, Proof Rules, and Hybrid Axiomatics}. Studia Logica 84(2), pp. 277--322 (2006).

\bibitem{hyb}
Blackburn, P., Tzakova, M.: \emph{Hybrid Completeness}. Logic Journal of the IGPL 4, pp. 625--650 (1998).

\bibitem{hybtemp}
Blackburn, P., Tzakova, M.: \emph{Hybrid languages and temporal logic}. Logic Journal of the IGPL 7, pp. 27--54 (1999).

\bibitem{mod}
Blackburn, P., Venema, Y., de Rijke, M.: Modal Logic. Cambridge University Press (2002).
  
%\bibitem{platzerhybrid}  
%Bohrer, B., Platzer, A.: {\emph A Hybrid, Dynamic Logic for Hybrid-Dynamic Information Flow}. In: LICS'18 Proceedings of the 33rd Annual ACM/IEEE Symposium on Logic in Computer Science, pp. 115--124 (2018).
  
%\bibitem{popl}
%Calcagno, C., Gardner, P., Zarfaty, U.: \emph{Context logic as modal logic: completeness and parametric inexpressivity}. In: POPL'07 Proceedings of the 34th annual ACM SIGPLAN-SIGACT symposium on Principles of programming languages, pp. 123--134 (2007).
  
%\bibitem{rosulics}  
%Chen, X., Ro\c su, G.: \emph{Matching mu-Logic}.  LICS'19. To appear. Technical report: http://hdl.handle.net/2142/102281 (2019).

%\bibitem{floyd}
%Floyd, R. W.: \emph {Assigning meanings to programs}. In: Proceedings of the American Mathematical Society Symposia on Applied Mathematics 19, pp. 19--31 (1967).
 
\bibitem{goranko}
Gargov, G., Goranko, V.: \emph{Modal logic with names}. Journal of Philosophical Logic 22, pp. 607--636 (1993).

\bibitem{goranko2}  
Goranko, V., Vakarelov, D.: \emph{Sahlqvist Formulas in Hybrid Polyadic Modal Logics}. Journal of Logic and Computation 11 (2001).
%
%\bibitem{goguen}
%Goguen, J., Malcolm, G.: Algebraic Semantics of Imperative Programs.  MIT Press (1996).


\bibitem{dynamic}
Harel, D., Tiuryn, J., Kozen, D.: Dynamic logic. MIT Press Cambridge (2000) 

  
  \bibitem{HHKR89}
Heering, J.,  Hendriks, P.R.H., Klint, P., Rekers, J., \emph{The syntax definition formalism SDF ---reference manual---}. ACM Sigplan Notices 24(11), pp. 43--75 (1989).

%\bibitem{hoare}
%Hoare, C. A. R.: \emph{An axiomatic basis for computer programming}. Communications of the ACM 12(10), pp. 576--580 (1969). 


%\bibitem{kracht}
%M. Kracht, Tools and Techniques in Modal Logic,  Studies in logic and the foundations of mathematics, vol. 142,  Elsevier (1999).

\bibitem{noi}
Leu\c stean, I., Moang\u a, N., \c Serb\u anu\c t\u a, T. F.: \textit{A many-sorted polyadic modal logic}. arXiv:1803.09709, submitted (2018).

\bibitem{noi2}
Leu\c stean, I., Moang\u a, N., \c Serb\u anu\c t\u a, T. F.:
\textit{Operational semantics and program verification using many-sorted hybrid modal logic}. arXiv:1905.05036 (2019)

\bibitem{noi3}
Leu\c stean, I., Moang\u a, N., \c Serb\u anu\c t\u a, T. F.:
\textit{From Hybrid Modal Logic to Matching Logic and Back}. arXiv:1909.00584 (2019)
%
%\bibitem{prior}
%P. Øhrstrøm and P. Hasle. A. N. Prior’s rediscovery of tense logic. Erkenntnis, 39, pp. 23–-50 (1993).

%\bibitem{platzerbook}
%Platzer, A.: Logical Foundations of Cyber-Physical Systems. Springer (2018).

\bibitem{plotkin}
Plotkin, G. D.: A Structural Approach to Operational Semantics (1981) Tech. Rep. DAIMI FN-19, Computer Science Department, Aarhus University, Aarhus, Denmark. (Reprinted with corrections in J. Log. Algebr. Program) 60-61, pp. 17--139 (2004).

\bibitem{rosu}
Ro\c su, G.: \emph{Matching logic}. In:  Logical Methods in Computer Science 13(4),lmcs:4153, pp. 1--61 (2017).



%\bibitem{separation}
%Reynolds, J. C.:\emph{Separation logic: A logic for shared mutable data structures}. In: Proceedings 17th Annual IEEE Symposium on Logic in Computer Science (2002).

%\bibitem{manys6}
%Schr\"{o}der, L., Pattinson, D.:
%\emph{Modular algorithms for heterogeneous modal logics via multi-sorted coalgebra}. In: Mathematical Structures in Computer Science 21(2) , pp. 235--266 (2011).

  
%\bibitem{manys2}
%Venema, Y.: \emph{Points, lines and diamonds: a two-sorted modal logic for projective planes}. In: Journal of Logic and Computation, pp. 601--621 (1999).
 
\end{thebibliography}

\end{document}